\documentclass[11pt,twoside,english]{article}
\usepackage{mathptmx}
\usepackage[T1]{fontenc}
\usepackage[utf8]{inputenc}
\usepackage[a4paper]{geometry}
\usepackage{color}
\usepackage{babel}
\usepackage{comment}
\usepackage{amsmath}
\usepackage{amsthm}
\usepackage{amssymb}
\usepackage{esint}
\usepackage{graphicx}
\usepackage{caption}
\usepackage{float}
\usepackage{subcaption}
\usepackage{upgreek}

\usepackage{adjustbox}
\usepackage{chngcntr}
\counterwithin{table}{section}

\usepackage[unicode=true,pdfusetitle,
 bookmarks=true,bookmarksnumbered=false,bookmarksopen=false,
 breaklinks=false,pdfborder={0 0 1},backref=false,colorlinks=true]
 {hyperref}

\makeatletter
\numberwithin{equation}{section}
\numberwithin{figure}{section}
\theoremstyle{plain}
\newtheorem{thm}{\protect\theoremname}
\theoremstyle{remark}
\newtheorem{rem}[thm]{\protect\remarkname}
\theoremstyle{plain}
\newtheorem{lem}[thm]{\protect\lemmaname}
\theoremstyle{plain}

\allowdisplaybreaks[4]
\date{}

\makeatother

\providecommand{\corollaryname}{Corollary}
\providecommand{\lemmaname}{Lemma}
\providecommand{\remarkname}{Remark}
\providecommand{\theoremname}{Theorem}

\begin{document}

\global\long\def\divg{{\rm div}\,}%

\global\long\def\curl{{\rm curl}\,}%

\global\long\def\rt{\mathbb{R}^{3}}%

\global\long\def\rn{\mathbb{R}^{n}}%

\global\long\def\dir{\mathcal{E}}%

\title{Random vortex dynamics via functional stochastic differential equations}
\author{Zhongmin Qian\thanks{Mathematical Institute, University of Oxford, OX2 6GG, England, and
Oxford Suzhou Centre for Advanced Research. Email: $\mathtt{qianz@maths.ox.ac.uk}$}, ~Endre Süli\thanks{Mathematical Institute, University of Oxford, OX2 6GG, England, and Oxford Centre for Nonlinear Partial Differential Equations. Email:
$\mathtt{endre.suli@maths.ox.ac.uk}$}, \ and\ Yihuang Zhang\thanks{Mathematical Institute, University of Oxford, OX2 6GG, England, and
Oxford Suzhou Centre for Advanced Research. Email: $\mathtt{yihuang.zhang@maths.ox.ac.uk}$}}
\maketitle
\begin{abstract}
In this paper we present a novel, closed, three-dimensional (3D)
random vortex dynamics system, which is equivalent to the Navier--Stokes
equations for incompressible viscous fluid flows. The new random vortex
dynamics system consists of a stochastic differential equation which is, 
in contrast with the two-dimensional random vortex dynamics equations, 
coupled to a finite-dimensional ordinary functional differential equation. 
This new random vortex system paves the way for devising
new numerical schemes (random vortex methods) for solving three-dimensional
incompressible fluid flow equations by Monte Carlo simulations. In
order to derive the 3D random vortex dynamics equations, we have developed
two powerful tools: the first is the duality of the conditional distributions
of a couple of Taylor diffusions, which provides a path space version
of integration by parts; the second is a forward type Feynman--Kac
formula representing solutions to nonlinear parabolic equations in
terms of functional integration. These technical tools and the underlying
ideas are likely to be useful in treating other nonlinear problems.
\end{abstract}


\section{Introduction}

The statement of the point vortex method may be traced back to the paper \cite{Rosenhead1932}, while the construction of effective computational methods by using
vortex motions for viscous fluid flows began with the important work  \cite{Chorin1973} (cf. also \cite{ChorinBernard1973} and
\cite{Chorin1979}) in which several important ideas, which became the core of vortex methods, were introduced to turn point-vortex approximations into powerful simulation tools.
While point-vortex approximations are not exactly equivalent to the underlying
fluid dynamics equations, and their convergence to true solutions of the equations
is hard to prove, a great deal of effort has been devoted to proving the convergence of various numerical schemes for the approximate solution of fluid flow equations based on vortex approximations.

The convergence of the random vortex method for two-dimensional (2D) incompressible viscous fluid flow was studied in \cite{Goodman1987}, and the rate of convergence was improved in \cite{Long1988}. Both authors utilized the fact that for 2D fluid flows vortex motions have the special feature that vortices are transported without experiencing nonlinear stretching.

The study of 3D inviscid fluid flows based on vortex dynamics has been an important discipline; one can find a review of advances in this field in \cite{MajdaandBertozzi2002} and  \cite{CottetandKoumoutsakos2000}, for example. 3D vortex methods for inviscid fluid flows were first presented by \cite{AndersonGreengard1985}, and the convergence of the method was subsequently also proved. In \cite{BealeMajda1982a, BealeMajda1982b}, the convergence of 3D vortex methods based on a Lagrangian frame was proved, and a new vortex method without the calculation of the Lagrangian frame was later proposed in \cite{Beale1986}.  In \cite{CottetGoodmanHou1991},  a vortex method without mollifying the singular Biot--Savart kernel was proposed.  For viscous flows certain numerical methods for computing solutions to fluid dynamics equations based on vortex dynamics were studied, for example, in \cite{knio1990numerical} and \cite{qian2021tracking}. However, the exact random vortex dynamics equations for the 3D incompressible Navier--Stokes equations have remained elusive. There are various stochastic representations of  solutions to the 3D Navier--Stokes equations in the existing literature.  In \cite{constantin2008stochastic}, a stochastic Lagrangian representation for the velocity of an incompressible fluid flow is obtained by making use of stochastic flows and Leray's projection. 
 In \cite{zhang2016stochastic} the existence and uniqueness of the solutions to the Navier--Stokes equation was studied  by using stochastic flow methods, and  a relationship between mild solutions  and stochastic representations has been investigated in \cite{olivera2021probabilistic}. In \cite{zhang2010stochastic}, stochastic flows associated with the Jacobian of the Taylor diffusion are used to study the solutions of the Navier--Stokes equation, while this approach is intended to be used for numerical algorithms. The authors of \cite{albeverio2010generalized, busnello2005probabilistic} proposed different approaches, and obtained McKean--Vlasov type stochastic differential equations, which however  involve a time-reversal in the coefficients. This formulation is therefore not amenable for numerical simulations. Other probabilistic representations are possible; for example, in \cite{ossiander2005probabilistic} solutions are expressed as functionals of a Markov process indexed by the nodes of a binary tree.

Over the last two decades several deterministic vortex methods were proposed to approximate the dynamics of viscous fluid flows. For example, a particle strength exchange method (PSE) was explored in \cite{degond1989weighted} and \cite{cottet1990particle}, where the authors approximate the Laplace operator via an integral operator and solve the related finite-dimensional ordinary differential equations. The Diffusion Velocity Method (DVM) was proposed and studied by  \cite{fronteau1984lie,ogami1991viscous} and \cite{beaudoin2003simulation}, where an artificial velocity field was incorporated into the dynamics in order to handle the diffusion of vorticity. In \cite{wang2021sharp} the dynamics of  Brownian motion type particles moving in chaotic flows are studied. More recently, machine learning techniques were applied in models of fluid flows; see, for example, \cite{kutz2017deep}, \cite{lee_you_2019} and \cite{gazzola2014reinforcement}. These methods have been useful for solving the 3D Navier--Stokes equations; they have, however, not revealed the underlying dynamics of vortices. 

In this paper we derive a system of random dynamical equations for the three-dimensional incompressible Navier--Stokes equations, which is an exact closure of Taylor's diffusion equation, the equation of motion for the ``imaginary'' Brownian fluid particles. The novel closure equations consist of a stochastic differential equation involving the law of its solution, coupled with an ordinary functional differential equation. The closed Taylor diffusion equations may then be used to design a random vortex method for the numerical approximation of the velocity field in three-dimensional incompressible viscous fluid flows by using a Monte Carlo method.

It is of course natural to explore the possibility of extending
the proposed approach to other important classes of fluid flows, such as wall-bounded flows and the dynamics of viscous fluid flows in boundary layers. At the time of writing, significant progress has been made in this direction. For example, we are able
to establish, with some subtle modifications, a forward type Feynman--Kac formula, and thus obtain similar, though more complicated, vortex dynamics for wall-bounded flows. The most subtle aspect in relation to wall-bounded flows with the approach proposed herein is handling the shear stress at the wall. The corresponding results will be published in a separate paper. The method presented here may also be applied to study other nonlinear problems by using diffusions constrained to domains with boundaries (cf. \cite{Venttsel1959} and \cite{StroockVaradhanBoundary},  for example). The reader may wish to consult \cite{GaoLiandXiao2018}, \cite{GaoRoquejoffre2022} and the references therein for a similar study. 

The paper is organised as follows. In Section 2 we describe the main ideas in our approach and establish the notation that will be used in subsequent sections. In Section 3 we introduce a useful tool, called the pinned diffusion measure, the conditional law of an irreversible diffusion given a terminal location, which enables us to perform integration-by-parts in the proof of our main result in Section 5. In Section 4 we develop another probabilistic tool, a forward type Feynman--Kac representation for solutions of the vorticity equation. In Section 5, by exploiting the tools developed in Sections 3 and 4, a closed and exact set of random vortex dynamics equations is obtained; its discretisation gives rise to a direct numerical scheme defined in Section 6. Numerical results and the assessment of the numerical method on model problems are discussed in Section 7.

\section{Several basic facts and the main ideas}

The random vortex method for two-dimensional (2D) viscous incompressible
fluid flows based on exact distributional stochastic differential
equations has been studied in \cite{Long1988}, whose approach makes
use of a special feature of the vorticity equations in two space dimensions,
that vortices in 2D fluid flows are simply transported
without nonlinear stretching of vortex motions.

Let $u=(u^{1},u^{2},u^{3})$ be the velocity field of an incompressible
fluid flow with viscosity $\nu>0$, and let $\omega:=\nabla\wedge u$ be
its vorticity. Then, the evolution of $u$ is governed by the incompressible Navier--Stokes equations
\begin{equation}
\frac{\partial u}{\partial t}+(u\cdot\nabla)u=\nu\Delta u-\nabla p\label{ns-m1}
\end{equation}
and
\begin{equation}
\nabla\cdot u=0.\label{ns-m2}
\end{equation}
The vorticity $\omega$ evolves according to the vorticity equation
\begin{equation}
\frac{\partial\omega}{\partial t}+(u\cdot\nabla)\omega=\nu\Delta\omega+(\omega\cdot\nabla)u,\label{vort-m1}
\end{equation}
where
\begin{equation}
\omega=\nabla\wedge u.\label{vort-m2}
\end{equation}
In the equations (\ref{vort-m1}), (\ref{vort-m2}) it is
the vorticity $\omega$ that is considered as the main fluid dynamic
variable, and the velocity $u$ is determined by the second equation,
\eqref{vort-m2}, together with the incompressibility condition \eqref{ns-m2}.
The fundamental difference between 2D fluid flows and 3D fluids flows
lies in the fact that for 2D flows (i.e., flows for which the third velocity component $u^{3}=0$, and the first two velocity components, $u^{1}$ and $u^{2}$, depend on the first two coordinates, $x^{1}$ and $x^{2}$, only) the nonlinear vorticity stretching term $(\omega\cdot\nabla)u$ appearing on the right-hand side of \eqref{vort-m1} vanishes identically; hence the vorticity $\omega=(0,0,\omega^{3})$, which is identified with the scalar function $\omega^{3}$, evolves according to the following parabolic equation: 
\begin{equation}
\frac{\partial\omega}{\partial t}+(u\cdot\nabla)\omega=\nu\Delta\omega,\quad\omega(\cdot,0)=\omega_{0}.\label{vort-m3}
\end{equation}
This crucial fact allows one to close the stochastic differential equations
of Taylor's Brownian motion $X$ (cf. \cite{Taylor1921,BealeMajda1981a,BealeMajda1982a,BealeMajda1982b}
and \cite{Long1988}): 
\[
\mathrm{d}X(t)=u(X(t),t)\,\mathrm{d}t+\sqrt{2\nu}\,\mathrm{d}B(t),\quad X(0)=x,
\]
issued from an initial state $x$, where $B$ is a three-dimensional
Brownian motion, to provide powerful random vortex methods for computing
2D fluid flows; cf. \cite{CottetandKoumoutsakos2000, Long1988}
and \cite{MajdaandBertozzi2002} for details. 

To describe our approach and the contributions in the present work,
we shall recall the key steps in the derivation of the 2D random vortex method
in more detail and establish the notation that we will use throughout
the paper. Let us begin by recalling several important mathematical
structures associated with the velocity vector field $u(x,t)$. We shall 
discuss these constructions for a general vector field $b(x,t)$
for reasons that will become clear in what follows.

For a time-dependent, bounded and Borel measurable vector field $b(x,t)$
(where $x\in\mathbb{R}^{d}$, $t\geq0$, and for the physically most relevant
case of interest to us here, the spatial dimension $d=3$), let 
\begin{equation}
L_{b(x,t)}:=\nu\Delta+b(x,t)\cdot\nabla\label{eq:form 2}
\end{equation}
which is a one-parameter family of second-order elliptic differential
operators on $\mathbb{R}^{d}$. While the assumption of boundedness
of $b$ is technical and may be replaced by other technical (weaker)
conditions, we will not pursue this issue further, as our key
aim in this paper is to derive a distributional stochastic differential
equation, which will be used to design random vortex methods. If there is no risk of confusion, $L_{b(x,t)}$ will be abbreviated as $L_{b}$, where
the arguments of $b$ are suppressed. Unless otherwise stated, we shall
assume that the drift vector field $b(x,t)$ is \emph{solenoidal}
(i.e., \emph{divergence-free}), that is, for every $t$, the divergence
$\nabla\cdot b(\cdot,t)$ vanishes identically (in the sense of distributions
on $\mathbb{R}^d$).
Under these assumptions, the formal adjoint operator $L_{b}^{\star}=L_{-b}$,
and therefore $L_{b}^{\star}$ is an elliptic operator of the same type as $L_b$. Let $\varGamma_{b}(x,t;\xi,\tau)$ (for $0\leq\tau<t$
and $\xi$, $x\in\mathbb{R}^{d}$) denote the fundamental solution to
the forward heat operator $L_{b}-\frac{\partial}{\partial t}$, and
$\varGamma_{b}^{\star}(x,t;\xi,\tau)$ (for $0\leq t<\tau$, $x$, $\xi\in\mathbb{R}^{d}$)
the fundamental solution to the backward heat operator $L_{b}^{\star}+\frac{\partial}{\partial t}$, 
respectively (cf. Chapter 1 in \cite{Friedman1964} for their definitions,
constructions and basic properties). Since $b$ is bounded,
both $\varGamma_{b}(x,t;\xi,\tau)$ and $\varGamma_{b}^{\star}(x,t;\xi,\tau)$
are positive and continuous in $x,\xi$ and $t>\tau\geq0$
(cf. \cite{Friedman1964} and \cite{Nash1958}, for example), and
$\varGamma_{b}(x,t;\xi,\tau)=\varGamma_{b}^{\star}(\xi,\tau;x,t)$
for $t>\tau\geq0$ and $x,\xi\in\mathbb{R}^{d}$. There are several
probabilistic structures associated with a vector field $b(x,t)$
too, and these constructions will play an important r{\^o}le in this paper.
Let $p_{b}(\tau,\xi,t,x)$ denote the transition probability density
function of the diffusion process with associated infinitesimal generator
$L_{b}$ (cf. \cite{IkedaandWatanabe1989} and \cite{StroockandVaradhan1979})
in the following sense. Suppose that $X(\xi,\tau,t)$ (where $\xi\in\mathbb{R}^{d}$
is the initial state at instance $\tau$, $\tau\leq t$) is
a diffusion process with associated infinitesimal generator $L_{b}$, i.e.,
$X$ is a (weak) solution to the stochastic differential equation
(SDE) 
\begin{equation}
\mathrm{d}X(\xi,\tau,t)=b(X(\xi,\tau,t),t)\,\mathrm{d}t+\sqrt{2\nu}\,\mathrm{d}B(t)\quad\textrm{ for }t\geq\tau\label{sde-m1}
\end{equation}
and 
\begin{equation}
X(\xi,\tau,s)=\xi\;\textrm{ for }s\leq\tau,\label{sde-m2}
\end{equation}
where $B$ is a $d$-dimensional Brownian motion on a probability
space $(\varOmega,\mathcal{F},\mathbb{P})$; $X(\xi,\tau,t)$ may
be called the Brownian fluid particle with velocity $b(x,t)$. In
the case when $\nu=0$, $t\rightarrow X(\xi,\tau,t)$ is just the
integral curve of the vector field $b(x,t)$ starting at $\xi \in \mathbb{R}^d$ at time $\tau$. Then, 
\begin{equation}
p_{b}(\tau,\xi,t,x)\,\mathrm{d}x=\mathbb{P}\left[X(\xi,\tau,t)\in \mathrm{d}x\right]\label{sde-m3}
\end{equation}
for any $t>\tau\geq0$. When $\tau=0$, we will simply
write $X(\xi,t)$ for $X(\xi,0,t)$ unless otherwise indicated. Since
$b(x,t)$ is divergence-free, $L_{b}^{\star}=L_{-b}$; therefore, 
\begin{equation}
p_{b}(\tau,\xi,t,x)=\varGamma_{-b}^{\star}(\xi,\tau;x,t)=\varGamma_{-b}(x,t;\xi,\tau)\label{def-r1}
\end{equation}
for $t>\tau\geq0$ and $\xi,x\in\mathbb{R}^{d}$. 

Up to now we have collected several constructions that are valid in
any number of space dimensions; we are now in a position to formulate the
random vortex dynamics equations in two dimensions.

For a 2D viscous flow with velocity $u(x,t)$, the vorticity $\omega=\nabla\wedge u$
solves the forward ``linear'' parabolic equation \eqref{vort-m3},
which may be written as 
\begin{equation}
\left(\frac{\partial}{\partial t}-L_{-u}\right)\omega=0,\quad\omega(x,0)=\omega_{0}(x),\label{vort-m3e}
\end{equation}
subject to appropriate boundary conditions at infinity. Hence,
\begin{equation}
\omega(x,t)=\int_{\mathbb{R}^{2}}\varGamma_{-u}(x,t;\xi,0)\,\omega_{0}(\xi)\, \mathrm{d}\xi\label{rep-m1}
\end{equation}
for $(x,t)\in\mathbb{R}^{2}\times[0,\infty)$, provided that $\omega_{0}(x)$
tends to zero sufficiently fast as $|x| \rightarrow \infty$. Thanks to \eqref{def-r1}
the equality \eqref{rep-m1} may be written as 
\begin{equation}
\omega(x,t)=\int_{\mathbb{R}^{2}}p_{u}(0,\xi,t,x)\,\omega_{0}(\xi)\, \mathrm{d}\xi,\label{rep-mp}
\end{equation}
where $p_{u}(0,\xi,t,x)$ is the probability density function of the Taylor
diffusion $X(\xi,t)$ with velocity $u$, which started at $\xi \in \mathbb{R}^2$ at $t=0$.
Now recall that the Taylor diffusion is the weak solution to the SDE
\begin{equation}
\mathrm{d}X(\xi,t)=u(X(\xi,t),t)\,\mathrm{d}t+\sqrt{2\nu}\,\mathrm{d}B(t),\quad X(\xi,0)=\xi\label{taylor-m2}
\end{equation}
for $t\geq0$. The idea in random vortex methods is to rewrite the SDE
\eqref{taylor-m2} as a closed stochastic differential equation
depending only on the initial datum $\omega_{0}$ and the distribution
of Taylor's diffusion $X$. To achieve this goal, one exploits the relations
$\omega=\nabla\wedge u$ and $\nabla\cdot u=0$, which yields
the Poisson equation $\Delta u=-\nabla\wedge\omega$, which, in turn, leads
to the Biot--Savart law. Under the assumption that both $u$ and $\omega$
decay to zero at infinity sufficiently fast, we have that
\begin{equation}
u(x,t)=\int_{\mathbb{R}^{2}}K(x-y)\,\omega(y,t)\, \mathrm{d}y,\label{u-rep-m1}
\end{equation}
where $K(z):=\frac{z^{\bot}}{|z|^{2}}$ is the 2D Biot--Savart kernel.
Together with \eqref{rep-mp} one may rewrite 
\begin{align*}
u(x,t) & =\int_{\mathbb{R}^{2}}\left[K(x-y)\int_{\mathbb{R}^{2}}p_{u}(0,\eta,t,y)\,\omega_{0}(\eta)d\eta\right]\,\mathrm{d}y\\
 & =\int_{\mathbb{R}^{2}}\left[\int_{\mathbb{R}^{2}}K(x-y)\,p_{u}(0,\eta,t,y)\,\mathrm{d}y\right]\omega_{0}(\eta)\,\mathrm{d}\eta\\
 & =\int_{\mathbb{R}^{2}}\mathbb{E}\!\left[K(x-X(\eta,t))\right]\omega_{0}(\eta)\,\mathrm{d}\eta.
\end{align*}
Thanks to this representation for the velocity $u(x,t)$ one is therefore
able to close the stochastic differential equation \eqref{taylor-m2}:
\begin{equation}
\mathrm{d}X(\xi,t)=\left(\left.\int_{\mathbb{R}^{2}}\mathbb{E}\!\left[K(x-X(\eta,t))\right]\omega_{0}(\eta)\,\mathrm{d}\eta\right|_{x=X(\xi,t)}\right)\mathrm{d}t+\sqrt{2\nu}\,\mathrm{d}B\label{sde-vm1}
\end{equation}
and 
\begin{equation}
X(\xi,0)=\xi\label{sde-vm2}
\end{equation}
for $(\xi,t)\in\mathbb{R}^{2}\times[0,\infty)$. The system of \eqref{sde-vm1},
\eqref{sde-vm2} is a stochastic differential equation dependent
on the law of its solution, i.e., it is a system of distributional
stochastic differential equations.

There are obvious advantages to using numerical methods based on the exact
stochastic dynamical system (\ref{sde-vm1}), (\ref{sde-vm2}). By employing
a discreterisation of the finite-dimensional integral appearing in \eqref{sde-vm1}
and by appealing to the strong law of large numbers for the expectation,
the closed system of stochastic differential equations (\ref{sde-vm1}), (\ref{sde-vm2}) leads to random vortex numerical schemes based on
Monte Carlo simulation techniques for 2D incompressible fluid flows. 

It would be helpful to formulate an exact random vortex dynamics system along
the same lines for viscous incompressible fluid flows in three space dimensions. There is, however, a substantial obstacle one has to overcome, caused by the appearance
of the nonlinear stretching term in the 3D vorticity equation. The nonlinear
stretching term plays a particularly important role in 3D; indeed, in its absence 
there is no turbulent motion. Let us explain the key difficulty we have to overcome.

We begin with the vorticity equations (\ref{vort-m1}), (\ref{vort-m2})
for incompressible fluid flow with viscosity $\nu>0$ and velocity
$u(x,t)$, which are rewritten as follows: 
\begin{equation}
\left(\frac{\partial}{\partial t}-L_{-u}\right)\omega=S\omega,\quad\omega(x,0)=\omega_{0}(x),\label{3Dvort-m1}
\end{equation}
where $S=(S_{j}^{i})$ is the symmetric velocity gradient tensor, with entries  $S_{j}^{i}:=\frac{1}{2}\left(\frac{\partial u^{i}}{\partial x^{j}}+\frac{\partial u^{j}}{\partial x^{i}}\right)$,
and $\omega:=\nabla\wedge u$ is now, unlike in the 2D setting, a three-component vector field. In three space dimensions,
the Biot--Savart law may be written as 
\begin{equation}
u(x,t)=\int_{\mathbb{R}^{3}}K(x-y)\wedge\omega(y,t)\, \mathrm{d}y,\label{u-rep-m3}
\end{equation}
where $K=(K^{1},K^{2},K^{3})$ is the singular vector kernel with components 
\begin{equation}
K^{i}(z):=-\frac{z^{i}}{4\pi|z|^{3}}\label{3d B_S kernel}
\end{equation}
for $i=1,2,3$. Because of the nonlinear stretching term, there is now no
explicit representation for the vorticity $\omega$, which is considered as
a solution to \eqref{3Dvort-m1}, in terms of the initial vorticity
$\omega_{0}$ and the fundamental solution of $L_{-u}$, although
the following implicit representation holds:
\begin{equation}
\omega(x,t)=\int_{\mathbb{R}^{3}}p_{u}(0,\eta,t,x)\,\omega_{0}(\eta)\, \mathrm{d}\eta+\int_{0}^{t}\int_{\mathbb{R}^{3}}p_{u}(s,\xi,t,x)\,S(\xi,s)\,\omega(\xi,s)\,\mathrm{d}\xi \,\mathrm{d}s;\label{3d-no-rep}
\end{equation}
this expression, however, is not useful for the purpose of closing 
the system of stochastic differential equations \eqref{taylor-m2} for Taylor diffusion. New ideas are therefore needed
in order to formulate an exact system of random vortex dynamics. 

A natural approach is to consider the symmetric velocity gradient tensor $S=(S_{j}^{i})$
as a multiplicative factor and apply the Feynman--Kac formula to the 
initial-value problem \eqref{3Dvort-m1} by constructing a gauge functional
from $S$ (cf. \cite{Feynman1948,Kac1950,Kac1959} for the linear case
and \cite{PardouxandPeng1990,Peng1991} for semilinear parabolic
equations). Since the elliptic operator $L_{-u(x,t)}$ is temporally
nonconstant and irreversible (which is canonical for flows
observed in nature), the Feynman--Kac formula for $\omega$ thus obtained
involves the distribution of a backward Taylor diffusion. Therefore
there are at least two substantial difficulties one needs to overcome.
First, this backward type Feynman--Kac formula cannot be used to close
the system of stochastic differential equations \eqref{taylor-m2}, as we have
to run a time-reversed Taylor diffusion in order to formulate the Feynman--Kac
formula. Second, there are currently no satisfactory numerical methods for the approximate solution of semilinear parabolic equations based on the nonlinear version of the (backward) Feynman--Kac formula.
In fact, existing methods appeal to numerical solutions of the
corresponding PDEs, which is then in conflict with our aim to compute numerical approximations of solutions to the Navier--Stokes equations via Monte Carlo simulations.

We overcome these difficulties by devising two powerful new tools,
which provide the means to close the system of stochastic differential equations
\eqref{taylor-m2} for Taylor diffusion with velocity $u(x,t)$.
Recall that in two space dimensions, in order to close the system of stochastic differential equations \eqref{taylor-m2} by eliminating the velocity $u$, 
two key steps are involved. The first is to devise a representation of the
vorticity $\omega$ by `solving' the vorticity equation, and the
second step is to apply the Biot--Savart law to obtain a representation
of the velocity $u$ in terms of $\omega$, then perform integration
by parts (in two space dimensions one has to do this at any fixed time $t$)
relying on the fact that $u$ is divergence-free, whereby
$L_{-u}=L_{u}^{\star}$). In our approach, in the three-dimensional
case, we devise a new representation formula, called a forward Feynman--Kac
formula, stated in \eqref{eq:F-K-f01} below, for the vorticity $\omega$
by using the Taylor diffusion rather than a time-reversed diffusion.
We are able to achieve this by using the fact that $u(x,t)$
is divergence-free, which implies that there is a convenient duality between
the $L_{u}$-diffusion and the $L_{-u}$-diffusion conditional on their
initial and final states; cf. Theorem \ref{Thm-Pinned} below.
We then close the system of stochastic differential equations \eqref{taylor-m2} by performing integration by parts at the path space level, -- a technique which also appears to be new.

\section{Pinned diffusion processes}

In this section we introduce the first technical tool in our study,
pinned diffusion measures and a time-reversal theorem. This is considered as a path-space version of the integration by parts technique. Let $\nu>0$
be a constant, let $b(x,t)$ be a Borel measurable and bounded vector
field depending on a parameter $t\geq0$, and let $L_{b}:=\nu\Delta+b\cdot\nabla$.
Let $\varOmega := C([0,\infty),\mathbb{R}^{d})$ be the path space
of all continuous paths in $\mathbb{R}^{d}$. The canonical coordinate
process on $\varOmega$ is the ordered family of coordinate mappings
$w\in\varOmega$, with value $w(t)$ at instance $t\geq0$, where
$t$ runs through $[0,\infty)$. Let $\mathcal{F}_{t}^{0}:=\sigma\{w(s):s\leq t\}$
be the smallest $\sigma$-algebra relative to which all coordinate mappings
$w\rightarrow w(s)$ are measurable for all $s\leq t$, and $\mathcal{F}^{0}:=\sigma\{w(s):s\geq0\}$.
We have thus turned the path space $C([0,\infty),\mathbb{R}^{d})$ into
a measurable space $(\varOmega,\mathcal{F}^{0})$. By the diffusion
process $X=(X_{t})_{t\geq0}$, with associated infinitesimal generator $L_{b}$,
we mean the Markov family of probability measures 
\[
\{\mathbb{P}_{b}^{\xi,\tau}:\xi\in\mathbb{R}^{d}\textrm{ and }\tau\geq0\},
\]
where $\mathbb{P}_{b}^{\xi,\tau}$ is the unique probability measure
on $(\varOmega,\mathcal{F}^{0})$ such that the finite-dimensional
joint distribution 
\[
\mathbb{P}_{b}^{\xi,\tau}\left[w\in\varOmega:w(s)=\xi\textrm{ for all }s\leq\tau,w(t_{1})\in \mathrm{d}x_{1},\ldots,w(t_{k})\in \mathrm{d}x_{k}\right]
\]
is given by 
\[
p_{b}(\tau,\xi,t_{1},x_{1})\,p_{b}(t_{1},x_{1},t_{2},x_{2})\,\cdots\, p_{b}(t_{k-1},x_{k-1},t_{k},x_{k})\,\mathrm{d}x_{1}\cdots \mathrm{d}x_{k}
\]
for all $k=0,1,2,\ldots$, and $\tau<t_{1}<\cdots<t_{k}$. Here, thanks
to the relation \eqref{def-r1}, one may identify the transition probability
density function $p_{b}(\tau,\xi,t,x)$ with the fundamental solution
$\varGamma_{-b}(x,t;\xi,\tau)$, which is positive and jointly Hölder
continuous. 

If $\tau\geq0$, $T>\tau$ are fixed, and $\xi,\eta\in\mathbb{R}^{d}$
are two given states, then $\mathbb{P}_{b}^{\xi,\tau\rightarrow\eta,T}$
denotes the distribution of the $L_{b}$-diffusion process $X=(X_{t})_{t\geq0}$
conditional on the event that $X_{s}=\xi$ for $s\leq\tau$ and $X_{s}=\eta$
for $s\geq T$. $\mathbb{P}_{b}^{\xi,\tau\rightarrow\eta,T}$ is again
a probability measure on the continuous path space $\varOmega$, called
the pinned diffusion measure starting from $\xi$ at instance $\tau$
and ending at state $\eta$ at the terminal time $T$, which can be
described as follows. It is known that the laws of a Markov
process conditional on the terminal values at some fixed time $T>0$
are also Markovian; see, for example, Section 14 in \cite{DellacherieandMeyerVolumeD} and \cite{ChungWalsh2005}. According to eq. (14.1) in \cite[pp. 161]{DellacherieandMeyerVolumeD},
we define
\begin{equation}
q_{b}(s,x,t,y):=\frac{p_{b}(s,x,t,y)\,p_{b}(t,y,T,\eta)}{p_{b}(s,x,T,\eta)}\label{cond-01}
\end{equation}
 for $\tau<s<t<T$ and $x,y\in\mathbb{R}^{d}$. It is then easy to
verify that 
\[
\int_{\mathbb{R}^{d}}q_{b}(s,x,t,y)\, \mathrm{d}y=1
\]
and, for $\tau<s<t<r<T$, one has that 
\[
\int_{\mathbb{R}^{d}}q_{b}(s,x,t,y)\,q_{b}(t,y,r,z)\, \mathrm{d}y=q_{b}(s,x,r,z).
\]
Hence, $\left\{ q_{b}(s,x,t,y):\tau\leq s<t<T\right\} $ defines
a temporal inhomogeneous Markov family. The $L_{b}$-pinned
diffusion measure or the $L_{b}$-diffusion bridge measure $\mathbb{P}_{b}^{\xi,\tau\rightarrow\eta,T}$
is by definition the conditional law of $\mathbb{P}_{b}^{\xi,\tau}\left[\,\cdot\,|w(T)=\eta\right]$.
It turns out that $\mathbb{P}_{b}^{\xi,\tau\rightarrow\eta,T}$ is
the unique probability measure on $(\varOmega,\mathcal{F}^{0})$ such
that the finite-dimensional marginal law 
\begin{equation}
\mathbb{P}_{b}^{\xi,\tau\rightarrow\eta,T}\left[w\in\varOmega:w(t_{0})\in \mathrm{d}x_{0},\ldots,w(t_{k})\in \mathrm{d}x_{k},\ldots,w(t_{n+1})\in \mathrm{d}x_{n+1}\right]\label{c-m1}
\end{equation}
is given by 
\begin{equation}
q(\tau,\xi,t_{1},x_{1})\prod_{k=1}^{n-1}q(t_{k},x_{k},t_{k+1},x_{k+1})\,q(t_{n},x_{n},T,\eta)\,\delta_{\xi}(\mathrm{d}x_{0})\,\mathrm{d}x_{1}\cdots\, \mathrm{d}x_{n}\,\delta_{\eta}(\mathrm{d}x_{n+1}),\label{c-m2}
\end{equation}
where 
\[
\tau=t_{0}<t_{1}<\cdots<t_{n}<t_{n+1}=T
\]
is any finite partition of $[\tau,T]$, and 
\begin{equation}
\mathbb{P}_{b}^{\xi,\tau\rightarrow\eta,T}\left[w\in\varOmega:w(s)=\xi\textrm{ and }w(t)=\eta\textrm{ for }s\leq\tau\textrm{ and }t\geq T\right]=1.\label{c-m3}
\end{equation}

By definition, if $\tau<s<T$ and if $f$ is a bounded Borel
measurable function $f$, we have that
\begin{align*}
\mathbb{P}_{b}^{\xi,\tau\rightarrow\eta,T}\left[f(w(s))\right] & =\int_{\mathbb{R}^{d}}p_{b}(\tau,\xi,s,y)\, f(y)\, \mathrm{d}y\\
 & =\int_{\mathbb{R}^{d}}\frac{p_{b}(s,y,T,\eta)}{p_{b}(\tau,\xi,T,\eta)}\, f(y)\, p_{b}(\tau,\xi,s,y)\, \mathrm{d}y\\
 & =\mathbb{P}_{b}^{\xi,\tau}\left[\frac{p_{b}(s,w(s),T,\eta)}{p_{b}(\tau,\xi,T,\eta)}\, f(w(s))\right],
\end{align*}
and we may therefore conclude that 
\begin{equation}
\left.\frac{\mathrm{d}\mathbb{P}_{b}^{\xi,\tau\rightarrow\eta,T}}{\mathrm{d}\mathbb{P}_{b}^{\xi,\tau}}\right|_{\mathcal{F}_{s}^{0}}=\frac{p_{b}(s,w(s),T,\eta)}{p_{b}(\tau,\xi,T,\eta)}\label{eq:den-01}
\end{equation}
for $s\in[\tau,T)$. 

If $\tau=0$ and $T>0$, then, via the canonical embedding, $\mathbb{P}_{b}^{\xi,\tau\rightarrow\eta,T}$
can be considered  (by an abuse of notation) as the probability
measure on $\varOmega$ of all continuous
paths in $\mathbb{R}^{d}$ with running time from $0$ to $T$. 

The duality between the laws of the $L_{b}$-diffusions and the laws
of the $L_{-b}$-diffusions registered in the next theorem plays an important rôle in the proposed new 3D random vortex method. 
\begin{thm}
\label{Thm-Pinned}Let $b^{T}(x,t)$ be a bounded, Borel measurable
and divergence-free vector field such that $b^{T}(x,t)=b(x,T-t)$
for all $x$ and $t\leq T$. Then, the associated pinned measures satisfy the
following duality relation: 
\begin{equation}
\mathbb{P}_{b}^{\xi,0\rightarrow\eta,T}=\mathbb{P}_{-b^{T}}^{\eta,0\rightarrow\xi,T}\circ\tau_{T},\label{duality-m1}
\end{equation}
where $\tau_{T}$ is the time-reversal operator at $T$, that is,
$\tau_{T}:\varOmega\rightarrow\varOmega$ is the mapping, which sends $w\in\varOmega$ to $\tau_{T}w \in \Omega$, where $\tau_{T}w(t)=w(T-t)$ for $t\in[0,T]$. 
\end{thm}

The assertion of this theorem may be reformulated as follows: suppose that $(X_{t})$ is an $L_{b}$-diffusion and $(Y_{t})$ is an $L_{-b^{T}}$-diffusion
on the probability space $(\varOmega,\mathcal{F},\mathbb{P})$; then,
\begin{equation}
\mathbb{E}\!\left[F(X_{\cdot})|X_{0}=\xi,X_{T}=\eta\right]=\mathbb{E}\!\left[F(Y_{T-\cdot})|Y_{0}=\eta,Y_{T}=\xi\right]\label{duality-m2}
\end{equation}
for any bounded or positive Borel measurable function $F$ on $\varOmega$.

To prove this theorem, we need the following simple fact concerning
fundamental solutions.
\begin{lem}
\label{lemma3.1}The following equality holds: 
\begin{equation}
p_{b^{T}}(T-t,x,T-\tau,\xi)=\varGamma_{b}(x,t;\xi,\tau),\label{eq:diff-backT1}
\end{equation}
and therefore 
\begin{equation}
p_{b}(s,x,t,y)=p_{-b^{T}}(T-t,y,T-s,x)\label{eq:rev-01}
\end{equation}
for $0\leq s<t\leq T$ and all $x,y\in\mathbb{R}^{d}$. 
\end{lem}

\begin{proof}
Recall that $p_{b^{T}}(t,x,\tau,\xi)$, for fixed $\tau$ and $\xi$,
as a function of $(x,t)$ is the unique solution to the backward equation
\begin{equation}
\left(L_{b^{T}}+\frac{\partial}{\partial t}\right)f=0,\label{eq:diff-bac-01-1}
\end{equation}
for $0\leq t<\tau$, such that $f(x,t)\uparrow\delta_{\xi}$ as $t\uparrow\tau$.
By performing the change of variables $t\rightarrow T-t$ and $\tau\rightarrow T-\tau$,
we have that
\begin{equation}
\left(L_{b}-\frac{\partial}{\partial t}\right)p_{b^{T}}(T-t,x,T-\tau,\xi)=0\label{eq:diff-bac-01-1-1}
\end{equation}
for $0\leq\tau<t\leq T$. Moreover, $p_{b^{T}}(T-t,x,T-\tau,\xi)\uparrow\delta_{\xi}$
as $t\downarrow\tau$. Thus, by the definition and the uniqueness
of $\varGamma_{b}(x,t;\xi,\tau)$, the equality \eqref{eq:diff-backT1}
directly follows. 
\end{proof}
\emph{Proof of Theorem \ref{Thm-Pinned}. } By definition,
\[
\mathbb{P}_{b}^{\eta,0\rightarrow\zeta,T}\left[w(t_{1})\in \mathrm{d}x_{1},\ldots,w(t_{k})\in \mathrm{d}x_{k}\right]
\]
is equal to the measure 
\[
q_{b}(0,\eta,t_{1},x_{1})\cdots q_{b}(t_{i-1},x_{i-1},t_{i},x_{i})\cdots q_{b}(t_{k},x_{k},T,\zeta)\, \mathrm{d}x_{1}\cdots\,  \mathrm{d}x_{k},
\]
which may be written as the following: 
\[
\frac{p_{b}(0,\eta,t_{1},x_{1})\cdots p_{b}(t_{i-1},x_{i-1},t_{i},x_{i})\cdots p_{b}(t_{k},x_{k},T,\zeta)}{p_{b}(0,\eta,T,\zeta)}\, \mathrm{d}x_{1}\cdots\,\mathrm{d}x_{k}.
\]
By Lemma \ref{lemma3.1}, the above finite-dimensional distribution
can be written as 
\[
\frac{p_{-b^{T}}(0,\zeta,T\!-\!t_{k},x_{k})\cdots p_{-b^{T}}(T\!-\!t_{i},x_{i},T\!-\!t_{i-1},x_{i-1}) \cdots p_{-b^{T}}(T\!-\!t_{1},x_{1},T,\eta)}{p_{-b^{T}}(0,\zeta,T,\eta)}\mathrm{d}x_{1}\cdots\, \mathrm{d}x_{k},
\]
which is the same as
\[
\mathbb{P}_{-b^{T}}^{\zeta,0\rightarrow\eta,T}\left[w(T-t_{k})\in \mathrm{d}x_{k},\cdots,w(T-t_{1})\in \mathrm{d}x_{1}\right].
\]
Now the conclusion follows, as a distribution is determined uniquely
by its family of finite-dimensional marginal distributions. This completes
the proof.\hfill $\Box$

\begin{rem}
The well-known setting of generalised integration by parts is the concept of dual (or adjoint) operators in functional analysis, and the concept of a pair of dual Markov processes \cite{BlumenthalGetoor1986}. The path space version (with respect to stationary measures of diffusion processes) was first formulated in \cite{tlyonswzheng1988} via forward and backward martingale decompositions (Lyons--Zheng's decomposition), and their approach was subsequently generalised to nonsymmetric diffusion processes; cf. \cite{tlyonsStoica1999} and the literature therein.
\end{rem}

\section{Forward Feynman--Kac's formula}

In this section we develop our second key tool, a forward Feynman--Kac
formula, which is the key in deriving the 3D random vortex dynamics equations.
We shall continue to assume that $b(x,t)$ is a bounded, Borel measurable and \emph{divergence-free} vector field, with $x \in \mathbb{R}^d$ and $t \geq 0$. 
\begin{thm}
\label{thm:F-K-01} Let $T>0$. Suppose that the vector field $f(x,t)$ is a strong solution
to the parabolic equation: 
\begin{equation}
\left(\frac{\partial}{\partial t}-L_{-b(x,t)}\right)f^{i}(x,t)=S_{j}^{i}(x,t)f^{j}(x,t) \quad \textrm{ in }\mathbb{R}^{d}\times(0,T],\label{eq:S-vot-eq}
\end{equation}
where $i=1,\ldots,d$, subject to the initial condition $f(x,0)=f_{0}$,
and $S_{j}^{i}$ are globally Lipschitz continuous for $i,j=1,\ldots,d$.
Then, 
\begin{equation}
f(x,T)=\int_{\mathbb{R}^{d}}f_{0}(\xi)\,p_{b}(0,\eta,T,x)\,\mathbb{P}\left[Q(0)|X(0)=\eta,X(T)=x\right]\mathrm{d}\xi\label{eq:F-K-f01},
\end{equation}
where 
\begin{equation}
\mathrm{d}X^{i}(t)=b^{i}(X(t),t)\,\mathrm{d}t+\sqrt{2\nu}\,\mathrm{d}B^{i}(t)\label{t-m1}
\end{equation}
and 
\begin{equation}
\frac{\mathrm{d}}{\mathrm{d}t}Q_{j}^{i}(t)=-Q_{k}^{i}(t)S_{j}^{k}(X(t),t),\quad Q_{j}^{i}(T)=\delta_{j}^{i}\label{t-m2}
\end{equation}
for $i,j=1,\ldots,d$, $B$ is a standard $d$-dimensional Brownian
motion on a probability space. 
\end{thm}

\begin{proof}
The solution of \eqref{eq:S-vot-eq} has a probabilistic representation
in terms of a functional integral, called the Feynman--Kac formula, which
however is in ``backward form''. Let $Y$ be the weak solution to
the following SDE: 
\begin{equation}
\mathrm{d}Y^{i}(t)=-b^{i}(Y(t),T-t)\,\mathrm{d}t+\sqrt{2\nu}\,\mathrm{d}B^{i}(t),\quad Y(0)=\eta,\label{eq:sde-bT1}
\end{equation}
where $t\in[0,T]$ and $i=1,\ldots,d$. That is to say, $Y$ is the
diffusion with infinitesimal generator $L_{-b^{T}}$, where $b^{T}(x,t)=b(x,T-t)$.
Define a gauge functional $Z=(Z_{j}^{i})$ by solving the system of ordinary
differential equations 
\begin{equation}
\frac{\mathrm{d}}{\mathrm{d}t}Z_{j}^{i}(t)=Z_{k}^{i}(t)\,S_{j}^{k}(Y(t),T-t),\quad Z_{j}^{i}(0)=\delta_{j}^{i},\label{eq:Q-v1}
\end{equation}
where $i,j=1,\ldots,d$, whose solution depends on $T$ and $\eta$. Then, by using Itô's formula,
$M^{i}(t)=Z_{k}^{i}(t)f^{k}(Y(t),T-t)$ is a martingale for $t\in[0,T]$:
\[
\mathrm{d}M^{i}(t)=Z_{k}^{i}(t)\,\nabla f^{k}(Y(t),T-t)\cdot \mathrm{d}B(t),
\]
which yields the following Feynman--Kac's formula:
\begin{equation}
f(\eta,T)=\mathbb{E}\!\left[Z(T)f_{0}(Y(T))\right].\label{eq:fey-kac-s1}
\end{equation}
Disintegrating the right-hand side by conditioning on $Y(T)=\zeta$,
we have that
\begin{align*}
f(\eta,T) & =\int_{\mathbb{R}^{d}}\mathbb{E}\!\left[Z(T)f_{0}(Y(T))|Y(T)=\zeta\right]\,\mathbb{P}\left[Y(T)\in \mathrm{d}\zeta\right]\\
 & =\int_{\mathbb{R}^{d}}\mathbb{E}\!\left[Z(T)|Y(0)=\eta,Y(T)=\zeta\right] f_{0}(\zeta)\,p_{-b^{T}}(0,\eta,T,\zeta)\, \mathrm{d}\zeta\\
 & =\int_{\mathbb{R}^{d}}\mathbb{E}\!\left[Z(T)|Y(0)=\eta,Y(T)=\zeta\right]f_{0}(\zeta)\,p_{b}(0,\zeta,T,\eta)\,\mathrm{d}\zeta.
\end{align*}
Let us compute the expectation $\mathbb{P}\left[Z(T)|Y(0)=\eta,Y(T)=\zeta\right]$.
To this end we employ the canonical setting, that is, $\varOmega$
is the path space of all continuous paths in $\mathbb{R}^{d}$ with
time duration $[0,T]$, and $\tau_{T}:\varOmega\rightarrow\varOmega$
is the time reversal map for a path, which maps $w(t)$ into $w(T-t)$. Let $\mathbb{P}_{-b^{T}}^{\eta,0}$
be the law of the diffusion started at $\eta$, with associated infinitesimal
generator $L_{-b^{T}}$. The version of $Z_{t}$ (written as $Z(w,t)$)
in this setting is the unique solution to the system of ordinary differential
equations 
\[
\frac{\mathrm{d}}{\mathrm{d}t}Z_{j}^{i}=Z_{k}^{i}S_{j}^{k}(w(t),T-t),\quad Z_{j}^{i}(w,0)=\delta_{j}^{i},
\]
which can be solved for every continuous path $w\in\varOmega$. By
definition,
\[
Z_{j}^{i}(w,t)=\delta_{j}^{i}+\int_{0}^{t}Z_{k}^{i}(w,s)\,S_{j}^{k}(w(s),T-s)\,\mathrm{d}s,
\]
and therefore 
\[
Z_{j}^{i}(\tau_{T}w,T-t)=\delta_{j}^{i}-\int_{T}^{t}Z_{k}^{i}(\tau_{T}w,T-s)\,S_{j}^{k}(w(s),s)\,\mathrm{d}s.
\]
Let $Q(\cdot,t)=Z(\cdot,T-t)\circ\tau_{T}$. Then $Q$ solves the
following system of ODEs: 
\[
\frac{\mathrm{d}}{\mathrm{d}t}Q_{j}^{i}(w,t)=-Q_{k}^{i}(w,t)\,S_{j}^{k}(w(t),t),\quad Q_{j}^{i}(w,T)=\delta_{j}^{i}
\]
for every continuous path $w\in\varOmega$. By Theorem 3.3,
\begin{align*}
\mathbb{P}_{-b^{T}}^{\eta,0}\left[Z(\cdot,T-t)|w(T)=\zeta\right] & =\mathbb{P}_{b}^{\zeta,0}\left[Z(\cdot,T-t)\circ\tau_{T}|w(T)=\eta\right]\\
 & =\mathbb{P}_{b}^{\zeta,0}\left[Q(\cdot,t)|w(T)=\eta\right],
\end{align*}
which yields the desired conclusion. 
\end{proof}

\section{The random vortex dynamics equations}

In this section we derive the 3D random vortex dynamics SDE. To this end, 
let $u(x,t)$ be the velocity field of a viscous incompressible fluid flow
(with viscosity $\nu>0$) without an external force applied, which moves
freely in $\mathbb{R}^{3}$, and let $\omega:=\nabla\wedge u$. We assume
that the velocity $u(x,t)$ and the vorticity $\omega(x,t)$ decay sufficiently
fast as $|x|\rightarrow\infty$ and are sufficiently
regular. Recall the vorticity equation 
\begin{equation}
\left(\frac{\partial}{\partial t}-L_{-u(x,t)}\right)\omega^{i}(x,t)=S_{j}^{i}(x,t)\,\omega^{j}(x,t),\label{eq:s-vot1}
\end{equation}
where $S$ is the symmetric velocity gradient tensor associated with $u$. Since
$u$ is divergence-free, it follows that $L_{-u}=L_{u}^{\star}$.

Our goal is to rewrite the stochastic equations of Taylor diffusion
with velocity $u(x,t)$ in closed form (i.e., independent of $u$):
\begin{equation}
\mathrm{d}X(\xi,t)=u(X(\xi,t),t)\,\mathrm{d}t+\sqrt{2\nu}\,\mathrm{d}B(t),\quad X(\xi,0)=\xi,\label{eq:X-e1}
\end{equation}
for all $\xi\in\mathbb{R}^{3}$, where $B$ is a three-dimensional
Brownian motion on some probability space $(\varOmega,\mathcal{F},\mathbb{P})$.
In three space dimensions we need to couple it with the following system of ordinary
differential equations to define the gauge functional $G(\xi,t,s)=G(s)=(G_{j}^{i}(s))$:
\begin{equation}
\frac{\mathrm{d}}{\mathrm{d}s}G_{j}^{i}(\xi,t,s)=-G_{k}^{i}(\xi,t,s)\,S_{j}^{k}(X(\xi,s),s),\quad G_{j}^{i}(\xi,t,t)=\delta_{j}^{i},\label{eq:X-e2}
\end{equation}
for $0\leq s\leq t$, where $S_{j}^{k}:=\frac{1}{2}(\frac{\partial u^{k}}{\partial x^{j}}+\frac{\partial u^{j}}{\partial x^{k}})$.
We want to devise a closed SDE system equivalent to \eqref{eq:X-e1}
and \eqref{eq:X-e2}, whose solutions give rise to the velocity field
$u(x,t)$.

By applying the forward Feynman--Kac formula \eqref{eq:F-K-f01} to \eqref{eq:s-vot1},
we obtain the following representation:
\begin{equation}
\omega^{k}(y,t)=\int_{\mathbb{R}^{3}}p_{u}(0,\xi,t,y)\,\omega_{0}^{j}(\xi)\,\mathbb{P}\!\left[G_{j}^{k}(\xi,t,0)|X(\xi,t)=y\right]\mathrm{d}\xi.\label{eq:F-K-f01-1}
\end{equation}
On the other hand, according to the Biot--Savart law 
\[
u^{i}(x,t)=-\int_{\mathbb{R}^{3}}\varepsilon^{ilk}\frac{x^{l}-y^{l}}{4\pi|x-y|^{3}}\,\omega^{k}(y,t)\, \mathrm{d}y.
\]
Substituting $\omega^{k}$ given by \eqref{eq:F-K-f01-1} into the Biot--Savart law, we then have that 
\begin{align*}
u^{i}(x,t) & =\int_{\mathbb{R}^{3}}\int_{\mathbb{R}^{3}}\varepsilon^{ilk}\frac{y^{l}-x^{l}}{4\pi|y-x|^{3}}\,\mathbb{E}\!\left[G_{j}^{k}(\xi,t,0)|X(\xi,t)=y\right]\omega_{0}^{j}(\xi)\,p_{u}(0,\xi,t,y)\,\mathrm{d}\xi \,\mathrm{d}y\\
 & =\int_{\mathbb{R}^{3}}\int_{\mathbb{R}^{3}}\varepsilon^{ilk}\frac{y^{l}-x^{l}}{4\pi|y-x|^{3}}\,\mathbb{E}\!\left[G_{j}^{k}(\xi,t,0)|X(\xi,t)=y\right]\mathbb{P}\left[X(\xi,t)\in \mathrm{d}y\right]\omega_{0}^{j}(\xi)\,\mathrm{d}\xi\\
 & =\int_{\mathbb{R}^{3}}\int_{\mathbb{R}^{3}}\mathbb{E}\!\left[\varepsilon^{ilk}\frac{y^{l}-x^{l}}{4\pi|y-x|^{3}}\,G_{j}^{k}(\xi,t,0)|X(\xi,t)=y\right]\mathbb{P}\left[X(\xi,t)\in \mathrm{d}y\right]\omega_{0}^{j}(\xi)\,\mathrm{d}\xi\\
 & =\int_{\mathbb{R}^{3}}\int_{\mathbb{R}^{3}}\!\mathbb{E}\!\left[\varepsilon^{ilk}\frac{X(\xi,t)^{l}-x^{l}}{4\pi|X(\xi,t)-x|^{3}}\,G_{j}^{k}(\xi,t,0)|X(\xi,t)=y\right]\mathbb{P}\left[X(\xi,t)\in \mathrm{d}y\right]\omega_{0}^{j}(\xi)\,\mathrm{d}\xi\\
 & =\int_{\mathbb{R}^{3}}\mathbb{E}\!\left[\varepsilon^{ilk}\frac{X(\xi,t)^{l}-x^{l}}{4\pi|X(\xi,t)-x|^{3}}\,G_{j}^{k}(\xi,t,0)\right]\omega_{0}^{j}(\xi)\,\mathrm{d}\xi,
\end{align*}
which expresses the velocity $u(x,t)$ in terms of the law of $(X,G)$.
To close the SDE system, we need to perform a similar computation
for the symmetric tensor $S=(S_{j}^{i})$. From the Biot--Savart law we obtain 
\begin{equation}
S_{j}^{k}(x,t)=\int_{\mathbb{R}^{3}}H_{j,i}^{k}(x-y)\,\omega^{i}(y,t)\,\mathrm{d}y,\label{eq:der-BS0}
\end{equation}
where the singular integral kernel $H_{j,i}^{k}$ is defined by
\begin{equation}
H_{j,i}^{k}(z):=\frac{3}{2}\frac{z^{l}}{4\pi|z|^{5}}\left(\varepsilon^{kli}z^{j}+\varepsilon^{jli}z^{k}\right).\label{eq:H-kernel0}
\end{equation}
Hence 
\begin{align*}
S_{j}^{k}(x,t) & =\int_{\mathbb{R}^{3}}\int_{\mathbb{R}^{3}}H_{j,i}^{k}(x-y)\,\mathbb{E}\!\left[G_{l}^{i}(\xi,t,0)|X(\xi,t)=y\right]p_{u}(0,\xi,t,y)\,\omega_{0}^{l}(\xi)\,\mathrm{d}y\,\mathrm{d}\xi\\
 & =\int_{\mathbb{R}^{3}}\int_{\mathbb{R}^{3}}H_{j,i}^{k}(x-y)\,\mathbb{E}\!\left[G_{l}^{i}(\xi,t,0)|X(\xi,t)=y\right]\mathbb{P}\left[X(\xi,t)\in \mathrm{d}y\right]\omega_{0}^{l}(\xi)\,\mathrm{d}\xi\\
 & =\int_{\mathbb{R}^{3}}\int_{\mathbb{R}^{3}}\mathbb{E}\!\left[H_{j,i}^{k}(x-y)\,G_{l}^{i}(\xi,t,0)|X(\xi,t)=y\right]\mathbb{P}\left[X(\xi,t)\in \mathrm{d}y\right]\omega_{0}^{l}(\xi)\,\mathrm{d}\xi\\
 & =\int_{\mathbb{R}^{3}}\int_{\mathbb{R}^{3}}\mathbb{E}\!\left[H_{j,i}^{k}(x-X(\xi,t))\,G_{l}^{i}(\xi,t,0)|X(\xi,t)=y\right]\mathbb{P}\left[X(\xi,t)\in \mathrm{d}y\right]\omega_{0}^{l}(\xi)\,\mathrm{d}\xi\\
 & =\int_{\mathbb{R}^{3}}\mathbb{E}\!\left[H_{j,i}^{k}(x-X(\xi,t))\,G_{l}^{i}(\xi,t,0)\right]\omega_{0}^{l}(\xi)\,\mathrm{d}\xi.
\end{align*}

Now we are able to introduce the following notions and notation,
which allow us to formulate a closed system of SDEs equivalent to
the Taylor diffusion equation.

Assume that $X(\xi,t)=(X^{i}(\xi,t))$ is a family of $\mathbb{R}^{d}$-valued
continuous semi-martingales, and suppose that $G(\xi,s,t)=\left(G_{j}^{i}(\xi,s,t)\right)$
(for each $\xi$ and $t>0$) is a family of $d\times d$ symmetric
matrix-valued continuous processes, differentiable in $s$, adapted
to the filtration generated by $X$; we then define 
\[
b_{(X,G)}^{i}(x,t):=\int_{\mathbb{R}^{3}}\mathbb{E}\!\left[\varepsilon^{ilk}\frac{X(\xi,t)^{l}-x^{l}}{4\pi|X(\xi,t)-x|^{3}}\,G_{\alpha}^{k}(\xi,t,0)\right]\omega_{0}^{\alpha}(\xi)\,\mathrm{d}\xi
\]
and 
\[
S_{(X,G),j}^{k}(x,t):=\int_{\mathbb{R}^{3}}\mathbb{E}\!\left[H_{j,\alpha}^{k}(x-X(\xi,t))\,G_{l}^{\alpha}(\xi,t,0)\right]\omega_{0}^{l}(\xi)\,\mathrm{d}\xi.
\]
What is important here is the fact that $b_{(X,G)}^{i}$ and $S_{(X,G),j}^{k}$
are determined by the joint distribution of $(X,G)$.

The random vortex system may now be formulated as follows: 
\begin{equation}
\mathrm{d}X^{k}(\xi,t)=b_{(X,G)}^{k}(X(\xi,t),t)\,\mathrm{d}t+\sqrt{2\nu}\,\mathrm{d}B^{k}(t),\quad X(\xi,0)=\xi\label{eq:dis-1}
\end{equation}
and 
\begin{equation}
\frac{\mathrm{d}}{\mathrm{d}s}G_{j}^{i}(\xi,t,s)=-G_{l}^{i}(\xi,t,s)\,S_{(X,G),j}^{l}(X(\xi,s),s),\quad G_{j}^{i}(x,t,t)=\delta_{j}^{i}\label{eq:dis-2}
\end{equation}
for $0\leq s\leq t$ and for all $t\geq0$, where $i,j,k=1,2,3$. 

Equations \eqref{eq:dis-1} and \eqref{eq:dis-2} together represent
a closed system of distributional stochastic differential equations.
The first equation \eqref{eq:dis-1} is a forward type distributional
SDE, and the second equation is a system of ordinary differential equations.
The system (\ref{eq:dis-1}), (\ref{eq:dis-2}) is not Markovian
but a system of ordinary functional differential equations (cf.
\cite{Hale1977}). 
\begin{lem}
If $b_{(X,G)}(x,t)$ is $C^{1}$, then 
\[
S_{(X,G),j}^{k}=\frac{1}{2}\left(\frac{\partial b_{(X,G)}^{k}}{\partial x^{j}}+\frac{\partial b_{(X,G)}^{j}}{\partial x^{k}}\right).
\]
\end{lem}
This result is of course a direct consequence of our definition.

\section{Direct numerical scheme}

Thanks to the closed form of our random vortex dynamical system (\ref{eq:dis-1}),
(\ref{eq:dis-2}), we are now able to devise a numerical scheme for
computing numerical approximations to solutions of the Navier--Stokes 
equations by performing Monte Carlo simulations.

A direct numerical scheme based on the closed random vortex dynamical
system (\ref{eq:dis-1}), (\ref{eq:dis-2}) can be constructed by approximating
the finite-dimensional integral and replacing the expectation by an
algebraic average (which can be justified by the strong law of large numbers)
appearing in this system. We will call this scheme a direct numerical
scheme based on the exact random vortex dynamics. Because this will be
beneficial in future studies, we break up the scheme into several steps.

We begin by recalling the singular kernels
\[
K_{k}^{i}(z):=-\varepsilon^{ilk}\frac{z^{l}}{4\pi|z|^{3}},\quad H_{j,i}^{k}(z):=\frac{3}{2}\frac{z^{l}}{4\pi|z|^{5}}\left(\varepsilon^{kli}z^{j}+\varepsilon^{jli}z^{k}\right).
\]

Let $h>0$ be the spatial lattice size in our numerical scheme, and let $x_{\boldsymbol{k}}$
be the lattice points defined by $x_{\boldsymbol{k}}:=h\boldsymbol{k}$ with $\boldsymbol{k}:=(k^{1},k^{2},k^{3})\in\mathbb{Z}^{3}$.
Let $X_{t}^{\boldsymbol{k}}:=X(x_{\boldsymbol{k}},t)$, $G_{j}^{\boldsymbol{k},l}(t,s):=G_{j}^{l}(x^{\boldsymbol{k}},t,s)$
(where $t\geq s\geq0$) and $\omega_{\boldsymbol{k}}^{j}:=\omega_{0}^{j}(x_{\boldsymbol{k}})$.
The first step is to replace the integrals on $\mathbb{R}^{3}$ with
the corresponding Riemann sums, so that
\[
b_{(X,G)}^{i}(x,t)=\int_{\mathbb{R}^{3}}\mathbb{E}\!\left[K_{k}^{i}(x-X(\xi,t))\,G_{j}^{k}(\xi,t,0)\right]\omega_{0}^{j}(\xi)\,\mathrm{d}\xi
\]
is approximated by the following discrete sum:
\[
h^{3}\sum_{\boldsymbol{k}\in\mathbb{Z}^{3}}\omega_{\boldsymbol{k}}^{j}\,\mathbb{E}\!\left[K_{l}^{i}(x-X_{t}^{\boldsymbol{k}})\,G_{j}^{\boldsymbol{k},l}(t,0)\right],
\]
and
\[
S_{(X,G),j}^{l}(x,t)=\int_{\mathbb{R}^{3}}\mathbb{E}\!\left[H_{j,\alpha}^{l}(x-X(\xi,t))\,G_{\beta}^{\alpha}(\xi,t,0)\right]\omega_{0}^{\beta}(\xi)\,\mathrm{d}\xi
\]
is approximated by the following discrete sum:
\[
h^{3}\sum_{\boldsymbol{k}\in\mathbb{Z}^{3}}\omega_{\boldsymbol{k}}^{\beta}\,\mathbb{E}\!\left[H_{j,\alpha}^{l}(x-X_{t}^{\boldsymbol{k}})\,G_{\beta}^{\boldsymbol{k},\alpha}(t,0)\right].
\]
Therefore, on a spatial mesh of spacing $h>0$, we may run the following approximate
random vortex system:
\begin{equation}
\mathrm{d}X_{t}^{\boldsymbol{i},i}=h^{3}\left.\sum_{\boldsymbol{k}\in\mathbb{Z}^{3}}\omega_{\boldsymbol{k}}^{j}\,\mathbb{E}\!\left[K_{l}^{i}(x-X_{t}^{\boldsymbol{k}})\,G_{j}^{\boldsymbol{k},l}(t,0)\right]\right|_{x=X_{t}^{\boldsymbol{i}}}\,\mathrm{d}t+\sqrt{2\nu}\,\mathrm{d}B_{t}\label{step1-dm1}
\end{equation}
with the initial condition $X_{0}^{\boldsymbol{i}}=x_{\boldsymbol{i}}$,
where $t\geq0$, and
\begin{equation}
\frac{\mathrm{d}}{\mathrm{d}s}G_{j}^{\boldsymbol{l},i}(t,s)=-G_{l}^{\boldsymbol{l},i}(t,s)\,h^{3}\left.\sum_{\boldsymbol{k}\in\mathbb{Z}^{3}}\omega_{\boldsymbol{k}}^{\beta}\,\mathbb{E}\!\left[H_{j,\alpha}^{l}(x-X_{s}^{\boldsymbol{k}})\,G_{\beta}^{\boldsymbol{k},\alpha}(s,0)\right]\right|_{x=X_{s}^{\boldsymbol{l}}}\label{step2-dm2}
\end{equation}
with the terminal value $G_{j}^{\boldsymbol{l},i}(t,t)=\delta_{j}^{i}$,
where $0\leq s\leq t$ for $t\geq0$.

In the next step we substitute the expectations by running $N$ independent
copies via independent 3D Brownian motions $B^{\rho}$, to define
the following stochastic differential equations
\begin{equation}
\mathrm{d}X_{t}^{(\boldsymbol{i},\rho),i}=\frac{h^{3}}{N}\sum_{\boldsymbol{k}\neq\boldsymbol{i}}\omega_{\boldsymbol{k}}^{\beta}\sum_{\sigma=1}^{N}\left[K_{\alpha}^{i}(X_{t}^{(\boldsymbol{i},\rho)}-X_{t}^{(\boldsymbol{k},\sigma)})\,G_{\beta}^{(\boldsymbol{k},\sigma),\alpha}(t,0)\right]\!\mathrm{d}t+\sqrt{2\nu}\,\mathrm{d}B_{t}^{\rho},\label{rv-m1}
\end{equation}
\begin{equation}
X_{0}^{(\boldsymbol{i},\rho)}=x_{\boldsymbol{i}},\quad\rho=1,\ldots,N,\label{rv-m2}
\end{equation}
coupled to
\begin{equation}
\hspace{-0.8mm}\frac{\mathrm{d}}{\mathrm{d}s}G_{j}^{(\boldsymbol{l},\lambda),i}(t,s)\!=\!-\frac{h^{3}}{N}\,G_{l}^{(\boldsymbol{l},\lambda),i}(t,s)\!\sum_{\boldsymbol{k}\neq\boldsymbol{l}}\omega_{\boldsymbol{k}}^{\beta}\sum_{\sigma=1}^{N}\!\left[H_{j,\alpha}^{l}(X_{s}^{(\boldsymbol{l},\lambda)}-X_{s}^{(\boldsymbol{k},\sigma)})\,G_{\beta}^{(\boldsymbol{k},\sigma),\alpha}(s,0)\right]\!,\label{rv-m3}
\end{equation}
\begin{equation}
G_{j}^{(\boldsymbol{l},\lambda),i}(t,t)=\delta_{j}^{i}\label{rv-m4}
\end{equation}
for $0\leq s\leq t$. The approximate velocity is therefore given
by
\[
\tilde{u}(x,t)=\frac{h^{3}}{N}\sum_{\boldsymbol{k}\neq\boldsymbol{i}}\omega_{\boldsymbol{k}}^{\beta}\sum_{\sigma=1}^{N}\left[K_{\alpha}^{i}(x-X_{t}^{(\boldsymbol{k},\sigma)})\,G_{\beta}^{(\boldsymbol{k},\sigma),\alpha}(t,0)\right].
\]

The rate of convergence of the sequence of solutions of the discretisation scheme (\ref{rv-m1}), (\ref{rv-m2}) towards the solution of the random vortex dynamics system (\ref{eq:dis-1}), (\ref{eq:dis-2}) is an interesting question to consider. Since the random vortex dynamics system involves a distributional stochastic differential equation, Wong--Zakai type convergence results for ordinary It\^o
stochastic differential equations are not applicable to the system (\ref{eq:dis-1}), (\ref{eq:dis-2}). What we can certainly say at this stage regardless is that the convergence rate will necessarily depend on the regularity of the solution to the Navier--Stokes equations. The convergence analysis of the scheme is both interesting and relevant, but is beyond the scope of the present paper and will be explored in a separate paper. 

\section{Simulation results}

According to equations (\ref{rv-m1})--(\ref{rv-m4}), the numerical algorithm is implemented as follows.

\begin{enumerate}

    \item We select initial data $(x_i,
    \omega_i)$, $i=1,2,\dots,n$. We scale $\omega_i$ by the factor $h^3$ here, so there is no need to include the constant $h^3$ in our numerical scheme.
    \item An iterative algorithm is used to solve the equations. First, the time interval $[0,T]$ is partitioned into $0=t_0\leq t_1\leq \dots\leq t_m=T$, where the partitions are of the same length $\frac{T}{m}$, and then we sample the $i.i.d.$ white noise $Z_{t_r}^{\rho}\sim N(0,\frac{T}{m}I_3)$ at $t_0,t_1,\dots, t_{m-1}$, and for each $1\leq \rho\leq N$, where $N$ is the number of independent copies.
    \item We aim to compute $(X_{t_r}^{(\boldsymbol{i},\rho)},G_{t_r,t_{r'}}^{(\boldsymbol{i},\rho)})$, where $0\leq r'\leq r \leq m, 1\leq \boldsymbol{i} \leq n,1\leq \rho\leq N$, such that it approximates the solution of (\ref{rv-m1}), (\ref{rv-m2}). That is,
\begin{align}\label{eq:numerical X}
\begin{aligned}
\hspace{-2mm}X_{t_{r+1}}^{(\boldsymbol{i},\rho),i}-X_{t_{r}}^{(\boldsymbol{i},\rho),i} &=\sum_{\boldsymbol{k}\neq\boldsymbol{i}}\frac{\omega_{\boldsymbol{k}}^{\beta}}{N}\sum_{\sigma=1}^{N}\left[K_{\alpha}^{i}(X_{t_r}^{(\boldsymbol{i},\rho)}-X_{t_r}^{(\boldsymbol{k},\sigma)})\,G_{\beta}^{(\boldsymbol{k},\sigma),\alpha}(t_r,0)\right]\Delta t+\sqrt{2\nu}\,Z_{t_r}^{\rho},\\X_{0}^{(\boldsymbol{i},\rho)}&=x_{\boldsymbol{i}},
\end{aligned}
\end{align}
\begin{align}\label{eq:numerical G}
\begin{aligned}
&G_{j}^{(\boldsymbol{l},\lambda),i}(t_r,t_{r'})-G_{j}^{(\boldsymbol{l},\lambda),i}(t_r,t_{r'-1})\\
&\qquad =-G_{l}^{(\boldsymbol{l},\lambda),i}(t_r,t_{r'})\sum_{\boldsymbol{k}\neq\boldsymbol{l}}\frac{\omega_{\boldsymbol{k}}^{\beta}}{N}\sum_{\sigma=1}^{N}\left[H_{j,\alpha}^{l}(X_{t_{r'}}^{(\boldsymbol{l},\lambda)}-X_{t_{r'}}^{(\boldsymbol{k},\sigma)})G_{\beta}^{(\boldsymbol{k},\sigma),\alpha}(t_{r'},0)\right]\Delta t
\end{aligned}
\end{align}
and $G_{j}^{(\boldsymbol{l},\lambda),i}(t_r,t_r)=\delta_{j}^{i}$ for $0\leq r\leq m. $

\item We observe that we are essentially computing $G(t,s)$ in the simplex $0\leq t_i\leq t_j\leq T$, where $0\leq i\leq j\leq m$. Thus, the iterations are constructed as follows: we start with the initial values $G_{j}^{(\boldsymbol{l},\lambda),i}(t_r,t_{r'})=\delta_j^i$, where $0\leq r'\leq r\leq m$. Based on this, $X$ is computed from \eqref{eq:numerical X}, and this pair $(X,G)$ is denoted by $(X,G)^{(0)}$. Iteratively,  given the pair $(X,G)^{(n)}$, the next iterate, $(X,G)^{(n+1)}$,  is defined by the following equations:
\begin{align}\label{eq:numerical2 G}
\begin{aligned}
&\bigg[G_{j}^{(\boldsymbol{l},\lambda),i}(t_r,t_{r'})-G_{j}^{(\boldsymbol{l},\lambda),i}(t_r,t_{r'-1})\bigg]^{(n+1)}
\\&=-G_{l}^{(\boldsymbol{l},\lambda),i}(t_r,t_{r'})^{(n+1)}\sum_{\boldsymbol{k}\neq\boldsymbol{l}}\frac{\omega_{\boldsymbol{k}}^{\beta}}{N}\sum_{\sigma=1}^{N}\left[H_{j,\alpha}^{l}(X_{t_{r'}}^{(\boldsymbol{l},\lambda)}-X_{t_{r'}}^{(\boldsymbol{k},\sigma)})\,G_{\beta}^{(\boldsymbol{k},\sigma),\alpha}(t_{r'},0)\right]^{(n)}\Delta t,
\end{aligned}
\end{align}
\begin{align}\label{eq:numerical2 X}
\begin{aligned}
&\bigg[X_{t_{r+1}}^{(\boldsymbol{i},\rho),i}-X_{t_{r}}^{(\boldsymbol{i},\rho),i}\bigg]^{(n+1)}  \\ &\quad = \sum_{\boldsymbol{k}\neq\boldsymbol{i}}\frac{\omega_{\boldsymbol{k}}^{\beta}}{N}\sum_{\sigma=1}^{N}\left[K_{\alpha}^{i}(X_{t_r}^{(\boldsymbol{i},\rho)}-X_{t_r}^{(\boldsymbol{k},\sigma)})\,G_{\beta}^{(\boldsymbol{k},\sigma),\alpha}(t_r,0)\right]^{(n+1)}\Delta t+\sqrt{2\nu}\,\mathrm{d}Z_{t_r}^{\rho}.
\end{aligned}
\end{align}

To be more precise, the equation \eqref{eq:numerical2 G} is solved backwards. In other words, we first let $[G_{j}^{(\boldsymbol{l},\lambda),i}(t_r,t_r)]^{(n+1)}=\delta_j^i$ for all $0\leq r \leq m$, and we then compute $[G_{j}^{(\boldsymbol{l},\lambda),i}(t_r,t_{r'})]^{(n+1)}$ based on \eqref{eq:numerical2 G} for each $r'$ from $r-1$ to $0$. Afterwards, based on \eqref{eq:numerical2 X}, we compute $[X_{t_{r}}^{(\boldsymbol{i},\rho,i)}]^{(n+1)}$ for each $r$ from $0$ to $m$.

\item We terminate the iterative process once the update from $(X,G)^{(n)}$ to $(X,G)^{(n+1)}$ is sufficiently small. In the following experiments, we terminated the iterative process once 
\[\sum_{0\leq r\leq m}\lVert X^{(n+1)}_{t_r}-X^{(n)}_{t_r}\rVert_X^2+\sum_{0\leq r\leq m}\lVert G^{(n+1)}(t_r,0)-G^{(n)}(t_r,0)\rVert_G^2\leq 10^{-7},
\]
where $\lVert \cdot\rVert_X^2$ means summing the squares through all possible $(\boldsymbol{i},\rho),i$ as the indices in \eqref{eq:numerical X} and $\lVert \cdot\rVert_G^2$ means summing the squares through all possible ${j},(\boldsymbol{l},\lambda),i$ as the indices in \eqref{eq:numerical G}.

\end{enumerate}

The computational complexity of one step in the iteration is $\mathcal{O}(m (Nn)^2)$, since at each time point, on the right-hand sides of \eqref{eq:numerical2 X} and \eqref{eq:numerical2 G}, we essentially compute $Nn$ matrices and then sum them up, and we need to execute this computation for all $1\leq \boldsymbol{i} \leq n,1\leq \rho\leq N$. Therefore, the overall complexity of this algorithm is $\mathcal{O}(mN^2n^2N_{\text{iter}})$, where $N_{\text{iter}}$ is the number of iterations that are performed.

In our actual experiments, because of the time discretisation, there is a possibility that $X_{t_r}^{(i,\rho)}$ might be close to $X_{t_r}^{(k,\sigma)}$  for some $i\neq k$, which could lead to an arithmetic overflow, because of the presence of singularities in the kernels $H$ and $K$ at zero. Therefore, mollifications of the singular kernels $H$ and $K$ are used in the  experiments reported below in order to avoid the appearance of overflow, and we sum through all $k$ instead of $k\neq l$ in \eqref{eq:numerical2 G} or $k\neq i$ in \eqref{eq:numerical2 X}, as there are no singularities in the mollified kernels. We ran the experiments on an AMD EPYC 7713 64-Core Processor, and the computations in \eqref{eq:numerical2 G} for $1\leq \boldsymbol{l} \leq n$, $1\leq \lambda\leq N$  and in  \eqref{eq:numerical2 X} for $1\leq \boldsymbol{i} \leq n,1\leq \rho\leq N$ were parallelised.

In the following simulations the time interval was in all cases (unless otherwise stated) partitioned into subintervals of length $0.02$ and the viscosity $\nu$ was set to $\frac{1}{2}$. The aim of the first simulation is to present a simple 3D visualisation of the computed streamlines of the flow emanating from a chosen initial condition. In our second and third experiment we compared our method with an exact solution and with another numerical scheme, respectively, to validate the proposed algorithm. For the sake of clarity, in the second and third experiment we chose to depict in the corresponding figures snapshots of the evolving velocity field, projected on the horizontal plane $z=0$.

In the first simulation (cf. Fig. \ref{figure:1}) the initial data are chosen as
\begin{equation*}
  (X_0^1, \omega_0^1) = ((0,0,0),(0,0,0.1)),\quad (X_0^2, \omega_0^2) = ((0,1,0),(0,0,-0.1))
\end{equation*}
and
\begin{equation*}
 (X_0^3, \omega_0^3) = ((0,0,1),(0,0.1,0)), 
\end{equation*}
each with 100 independent copies. Because of the diffusion, additional swirling eddies appear in the flow, and the magnitude of the velocity is diffused into the whole flow.

In the second simulation, we assess our numerical scheme by considering the  known evolution of the Lamb--Oseen vortex, a line vortex that decays because of the presence of viscosity. Although the Lamb--Oseen vortex is typically regarded as a 2D solution with initial vorticity $\delta(x)\delta(y)$ for $(x,y)\in \mathbb{R}^2$, we consider it here as a 3D vortex solution with initial vorticity $(0,0,\delta(x)\delta(y))$ for $(x,y,z)\in \mathbb{R}^3$. Then, the exact solution of the resulting 3D Lamb--Oseen vortex evolution problem, with $\nu=\frac{1}{2}$, is given by
\begin{equation}
    u(x,y,z,t) = \frac{1}{2\pi}\frac{(-y,x,0)}{x^2+y^2}(1-\mathrm{e}^{-\frac{x^2+y^2}{2t}}),\qquad (x,y,z) \in \mathbb{R}^3,\; t>0.
\end{equation}
 In our numerical simulations the evolution of the computed solution is initialized at $x_k=(0,0,k/2)$ with $\omega_k=(0,0,1/2)$ for $-20\leq k\leq 20$, and we consider the velocity field projected on the plane  $z=0$ when $t=0.1$ (cf. Fig. \ref{figure:2}). It is clear from the figures that the resemblance between the exact velocity field depicted in {Fig.~\ref{figure:2} (a)}
and the computed velocity field with $N=1$, $20$ and $100$ independent copies, shown in Fig. \ref{figure:2} (b), (c), (d), respectively, improves as $N$ increases. To compare the computed velocity fields with the exact velocity field more precisely, a lattice consisting of 25 points in the horizontal plane, $z=0$, was chosen; the values at these points of the exact and approximate velocity fields projected on the  plane $z=0$ are shown in Table \ref{tab:1}. In addition, we have calculated the error between the exact solution and the approximate solution in the discrete counterpart of the $L^1_{[0,1]^3}$ norm. More precisely, we sum $|u_{\text{approx}}- u_{\text{exact}}|$ over the lattice $(\frac{i}{10},\frac{j}{10},\frac{k}{10})_{-10\leq i,j,k\leq 9}$ and then multiply the resulting sum by $0.1^3$, where $u_{\text{approx}}$ is the velocity field calculated using our scheme, while $u_{\text{exact}}$ is the exact velocity field. In Table 7.1, only the first two coordinates of the 25 lattice points and the first two components of the three-component vectors of the exact and approximate velocity fields at the 25 lattice points are shown. The third component of the exact velocity field is identically $0$; for the values of the approximate velocity fields at the lattice points, listed in the third, fourth and fifth columns of the table, the third component is not exactly $0$, but was in all cases less than $10^{-4}$.

We apply the numerical scheme, slightly modified,   instead of \eqref{rv-m1} and \eqref{rv-m3}. More precisely we use the following scheme:
\begin{equation}
\mathrm{d}X_{t}^{(\boldsymbol{i},\rho),i}=\frac{h^{3}}{N}\sum_{\boldsymbol{k}\neq\boldsymbol{i}}\omega_{\boldsymbol{k}}^{\beta}\sum_{\sigma=1}^{N}\left[K_{\alpha}^{i}(X_{t}^{(\boldsymbol{i},\rho)}-X_{t}^{(\boldsymbol{k},\sigma)})\,G_{\beta}^{(\boldsymbol{k},\sigma),\alpha}(t,0)\right]\!\mathrm{d}t+\sqrt{2\nu}\,\mathrm{d}B_{t}^{i,\rho},
\end{equation}
\begin{equation}
\hspace{-0.8mm}\frac{\mathrm{d}}{\mathrm{d}s}G_{j}^{(\boldsymbol{l},\lambda),i}(t,s)\!=\!-\frac{h^{3}}{N}\,G_{l}^{(\boldsymbol{l},\lambda),i}(t,s)\!\sum_{\boldsymbol{k}\neq\boldsymbol{l}}\omega_{\boldsymbol{k}}^{\beta}\sum_{\sigma=1}^{N}\!\left[H_{j,\alpha}^{l}(X_{s}^{(\boldsymbol{l},\lambda)}-X_{s}^{(\boldsymbol{k},\sigma)})\,G_{\beta}^{(\boldsymbol{k},\sigma),\alpha}(s,0)\right].
\end{equation}
There are $N\rho$ independent Brownian motions instead of $\rho$ independent Brownian motions for the whole system and $X^{i,\rho}$ are driven by independent Brownian motions. A numerical algorithm is developed based on this system by time discretization. The numerical experiment with this scheme helps avoid using a large number of copies but still gives a stable result.  This scheme however reduces the computational costs. 

The scheme is  tested for the Taylor--Green vortex with the Reynolds number set to 1600. The initial velocity has components $\cos(x)\sin(y)\sin(z)$, $-\sin(x)\cos(y)\sin(z)$ and $0$. The simulation is initialized with the lattice points $(\frac{\pi}{16}i,\frac{\pi}{16}j,\frac{\pi}{16}k)_{-16\leq i,j,k\leq 48}$, and $N=1$. We compare the evolution of the velocity field projected on the plane $z=0.1$. We see that the simulation results resemble the real solution quite well (cf. Fig. \ref{figure:new1}). In order to explore the improvement resulting from reduction of the step size, we conducted another simulation with the same initialization but with step size $\Delta t=0.01$ instead of $0.02$ (cf. Fig. \ref{figure:new1}). In Table \ref{tab:2}, a lattice consisting of 25 points on this plane was chosen, and the values at these points of the exact and approximate velocity fields projected on the  plane $z=0.1$  are shown. In addition, we have calculated the $L^1_{[\frac{\pi}{2},\frac{3\pi}{2}]^2}$ norm error between the exact solution and the approximate solutions. It is clear that the accuracy of the approximation is improved by reducing the step size $\Delta t$.

Next, the scheme is tested in the case of artificial homogeneous isotropic turbulence. The turbulence is initiated by a solenoidal velocity field
\begin{equation}
\begin{aligned}
u(x,y,z) = (&\cos(x)\sin(y)\sin(z)+\cos(x)\sin(y)\sin(2z),\\&-\sin(x)\cos(y)\sin(z)+\sin(x)\cos(y)\sin(2z),\\&-\sin(x)\sin(y)\cos(2z)).
\end{aligned}
\end{equation}
We compare our numerical results with the results obtained by DNS. The DNS is based on a pseudo-spectral Fourier--Galerkin method \cite{canuto1988spectral}. The  code for the implementation of DNS is based on the algorithm suggested in \cite{orszag1972numerical,mortensen2016high}. The Reynolds number is  1600 as in the numerical verification in \cite{mortensen2016high}, so that the viscosity $\nu$ is $\frac{1}{1600}$, accordingly. The simulation is initialized using  the lattice points $(\frac{\pi}{16}i,\frac{\pi}{16}j,\frac{\pi}{16}k)_{-16\leq i,j,k\leq 48}$ and $N=1$. We compare the evolution of the velocity field projected onto the plane $z=0.1$.  The simulation and the DNS solution produce very similar results (cf. Fig. \ref{figure:new2}).

\begin{figure}[ht]
\centering
    \begin{subfigure}{0.3\textwidth}
        \centering
        \includegraphics[width=1\linewidth]{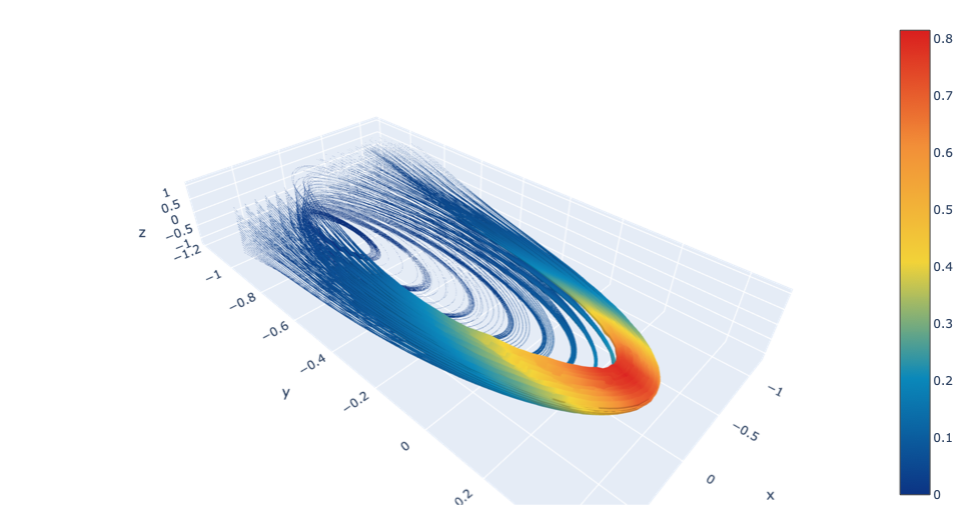}
        \caption{$t=0$} 
    \end{subfigure} 
    \begin{subfigure}{0.3\textwidth}
        \centering
        \includegraphics[width=1\linewidth]{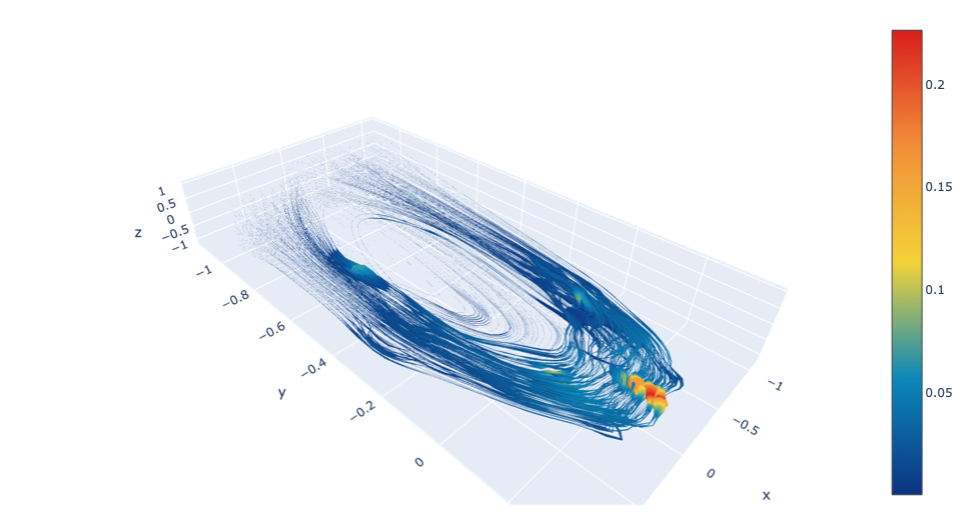}
        \caption{$t=0.1$}
    \end{subfigure}
    \begin{subfigure}{0.3\textwidth}
        \centering
        \includegraphics[width=1\linewidth]{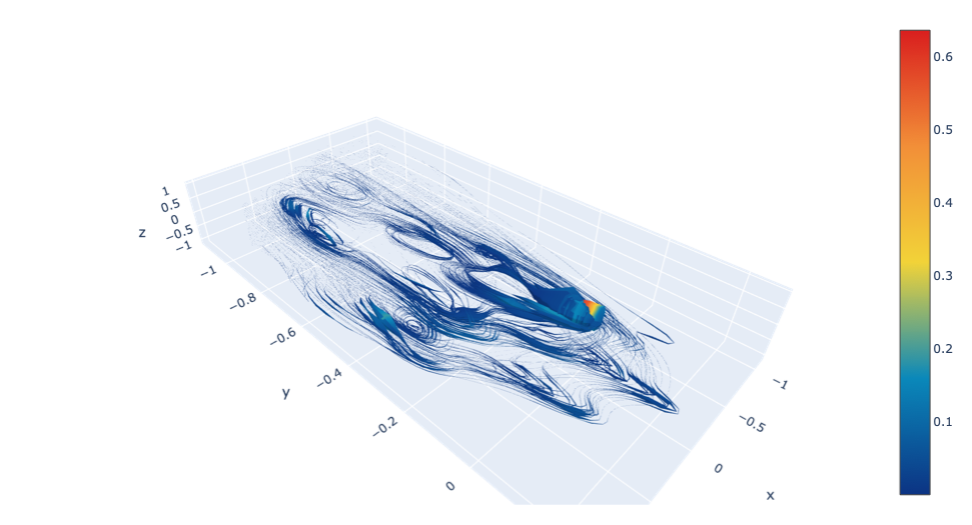}
        \caption{$t=0.2$} 
    \end{subfigure}
    \caption{Computed streamlines of the flow flow at $t=0$,  $0.1$, $0.2$ in the first experiment.}
\label{figure:1}
\end{figure}

\begin{figure}[ht]
\centering
    \begin{subfigure}{0.23\textwidth}
        \centering
        \includegraphics[width=1\linewidth]{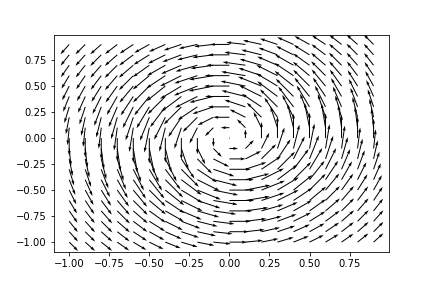}
        \caption{}
    \end{subfigure}
    \begin{subfigure}{0.23\textwidth}
        \centering
        \includegraphics[width=1\linewidth]{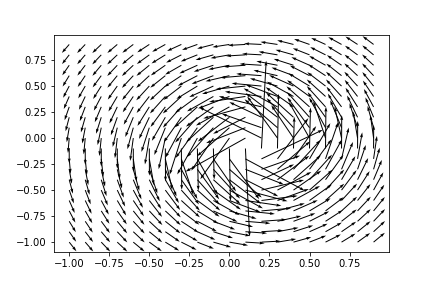}
        \caption{}
        
    \end{subfigure}
    \begin{subfigure}{0.23\textwidth}
        \centering
        \includegraphics[width=1\linewidth]{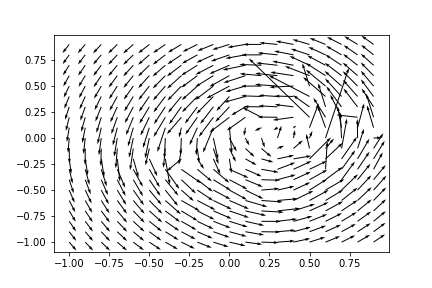}
        \caption{}
    \end{subfigure}
    \begin{subfigure}{0.23\textwidth}
           \centering
        \includegraphics[width=1\linewidth]{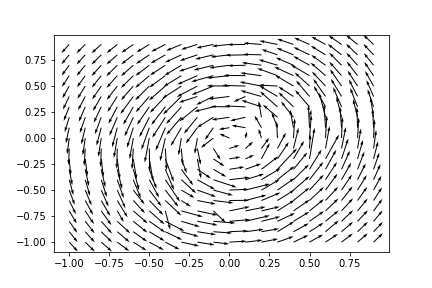}
        \caption{}
    \end{subfigure}
    \caption{The velocity field of the Lamb--Oseen vortex at $t=0.1$: (a) the exact solution; and the numerical solution computed with, respectively, (b) $N=1$, (c) $N=20$, (d) $N=100$ independent copies.}
\label{figure:2}
\end{figure}

\begin{table}[h]
\centering
\resizebox{0.8\textwidth}{!}{%
\begin{tabular}{|c|c|c|c|c|}
\hline
Point      & Exact~value        & $N=1$           & $N=20$          & $N=100$         \\ \hline
($-1.0$, $-1.0$) & ($0.08$, $-0.08$)  & ($0.06$,$-0.08$)  & ($0.06$,$-0.08$)  & ($0.08$, $-0.08$)  \\ \hline
($-0.5$, $-1.0$) & ($0.13$, $-0.06$)  & ($0.11$, $-0.08$)  & ($0.10$, $-0.07$)   & ($0.13$, $-0.07$)  \\ \hline
($0.0$, $-1.0$)  & ($0.16$, $0.00$)    & ($0.17$, $-0.03$)  & ($0.14$, $-0.03$)  & ($0.15$, $0.01$)   \\ \hline
($0.50$, $-1.0$)  & ($0.13$, $0.06$)   & ($0.15$, $0.06$)   & ($0.14$, $0.04$)   & ($0.12$, $0.07$)   \\ \hline
($1.00$, $-1.0$)  & ($0.08$, $0.08$)   & ($0.09$, $0.08$)   & ($0.10$, $0.07$)    & ($0.08$, $0.08$)   \\ \hline
($-1.0$, $-0.5$) & ($0.06$, $-0.13$)  & ($0.04$, $-0.12$)  & (0.05, $-0.11$)  & ($0.05$, $-0.13$)  \\ \hline
($-0.5$, $-0.5$) & ($0.15$, $-0.15$)  & ($0.11$, $-0.17$)  & ($0.11$, $-0.13$)  & ($0.11$, $-0.14$)  \\ \hline
($0.0$, $-0.5$)  & ($0.23$, $0.00$)    & ($0.30$, $-0.16$)   & ($0.16$, $-0.09$)  & ($0.17$, $0.01$)   \\ \hline
($0.5$, $-0.5$)  & ($0.15$, $0.15$)   & ($0.21$, $0.18$)   & ($0.16$, $0.10$)    & ($0.13$, $0.14$)   \\ \hline
($1.0$, $-0.5$)  & ($0.06$, $0.13$)   & ($0.07$, $0.15$)   & ($0.10$, $0.13$)    & ($0.06$, $0.12$)   \\ \hline
($-1.0$, $0.0$)  & ($-0.00$, $-0.16$)  & ($-0.01$, $-0.13$) & ($-0.01$, $-0.13$) & ($-0.02$, $-0.14$) \\ \hline
($-0.5$, $0.0$)  & ($-0.00$, $-0.23$)  & ($-0.04$, $-0.24$) & ($-0.04$, $-0.18$) & ($-0.04$, $-0.22$) \\ \hline
($0.0$, $0.0$)   & ($-0.00$, $0.00$)    & ($-0.26$, $-0.37$) & ($0.06$, $-0.12$)  & ($-0.09$, $0.04$)  \\ \hline
($0.5$, $0.0$)   & ($-0.00$, $0.23$)   & ($-0.18$, $0.32$)  & ($0.02$, $0.06$)   & ($0.02$, $0.21$)   \\ \hline
($1.0$, $0.0$)   & ($-0.00$, $0.16$)   & ($-0.02$, $0.18$)  & ($-0.03$, $0.24$)  & ($0.00$, $0.15$)    \\ \hline
($-1.0$, $0.5$)  & ($-0.06$, $-0.13$) & ($-0.06$, $-0.10$)  & ($-0.05$, $-0.10$)  & ($-0.06$, $-0.11$) \\ \hline
($-0.5$, $0.5$)  & ($-0.15$, $-0.15$) & ($-0.12$, $-0.13$) & ($-0.10$, $-0.12$)  & ($-0.13$, $-0.13$) \\ \hline
($0.0$, $0.5$)   & ($-0.23$, $0.00$)   & ($-0.24$, $-0.07$) & ($-0.21$, $-0.11$) & ($-0.20$, $0.00$)    \\ \hline
($0.5$, $0.5$)   & ($-0.15$, $0.15$)  & ($-0.19$, $0.11$)  & ($-0.06$, $0.20$)   & ($-0.13$, $0.12$)  \\ \hline
($1.0$, $0.5$)   & ($-0.06$, $0.13$)  & ($-0.09$, $0.12$)  & ($-0.10$, $0.13$)   & ($-0.06$, $0.13$)  \\ \hline
($-1.0$, $1.0$)  & ($-0.08$, $-0.08$) & ($-0.07$, $-0.07$) & ($-0.06$, $-0.07$) & ($-0.07$, $-0.07$) \\ \hline
($-0.5$, $1.0$)  & ($-0.13$, $-0.06$) & ($-0.10$, $-0.06$)  & ($-0.10$, $-0.07$)  & ($-0.11$, $-0.06$) \\ \hline
($0.0$, $1.0$)   & ($-0.16$, $0.00$)   & ($-0.14$, $-0.02$) & ($-0.14$, $-0.04$) & ($-0.14$, $-0.01$) \\ \hline
($0.5$, $1.0$)   & ($-0.13$, $0.06$)  & ($-0.13$, $0.04$)  & ($-0.14$, $0.04$)  & ($-0.13$, $0.05$)  \\ \hline
($1.0$, $1.0$)   & ($-0.08$, $0.08$)  & ($-0.09$, $0.07$)  & ($-0.10$, $0.07$)   & ($-0.08$, $0.08$)  \\ \hline \hline
    $L^1_{[-1,1]^3}$ norm error  &  ~             &$\mbox{Error $=0.91$}$~~          &$\mbox{Error $=0.66$}$~~          &$\mbox{Error $=0.19$}$~~          \\ \hline
\end{tabular}%
}
\caption{The exact values of the velocity field in the Lamb--Oseen vortex problem at $t=0.1$ (second column) at 25 lattice points (first column), and the values of the computed velocity field at those points, with $N = 1$ (third column), $N=20$ (fourth column), and $N=100$ (fifth column) independent copies, corresponding to the images depicted in Fig. \ref{figure:2}. The discrete $L^1_{[-\frac,1]^3}$ norm error based on the $20^3$ lattice points $(\frac{i}{10},\frac{j}{10},\frac{k}{10})_{-10\leq i,j,k\leq 9}$ is shown in the last row of the table.}

\label{tab:1}
\end{table}
\bigskip

\begin{table}[]
\centering
\resizebox{0.8\textwidth}{!}{
\begin{tabular}{|c|c|c|c|}
\hline
Point      & Exact~value                & $\Delta t =0.02$          & $\Delta t =0.01$         \\ \hline
( $1.57$ , $1.57$) & ( $-0.00$ , $-0.00$) & ( $0.00$ , $-0.00$) & ( $0.00$ , $-0.01$) \\ \hline
( $2.20$ , $1.57$) & ( $0.05$ , $-0.00$) & ( $0.04$ , $-0.00$) & ( $0.03$ , $-0.01$) \\ \hline
( $2.83$ , $1.57$) & ( $-0.04$ , $-0.00$) & ( $-0.08$ , $-0.01$) & ( $-0.07$ , $-0.01$) \\ \hline
( $3.46$ , $1.57$) & ( $-0.17$ , $-0.00$) & ( $-0.15$ , $-0.00$) & ( $-0.18$ , $-0.01$) \\ \hline
( $4.08$ , $1.57$) & ( $-0.17$ , $-0.00$) & ( $-0.15$ , $-0.00$) & ( $-0.17$ , $-0.00$) \\ \hline
( $1.57$ , $2.20$) & ( $-0.00$ , $0.17$) & ( $-0.00$ , $0.15$) & ( $0.00$ , $0.17$) \\ \hline
( $2.20$ , $2.20$) & ( $0.06$ , $0.15$) & ( $0.06$ , $0.13$) & ( $0.04$ , $0.15$) \\ \hline
( $2.83$ , $2.20$) & ( $-0.02$ , $0.12$) & ( $-0.13$ , $0.14$) & ( $-0.04$ , $0.11$) \\ \hline
( $3.46$ , $2.20$) & ( $-0.15$ , $0.09$) & ( $-0.15$ , $0.08$) & ( $-0.16$ , $0.08$) \\ \hline
( $4.08$ , $2.20$) & ( $-0.15$ , $0.06$) & ( $-0.14$ , $0.04$) & ( $-0.15$ , $0.04$) \\ \hline
( $1.57$ , $2.83$) & ( $-0.00$ , $0.17$) & ( $-0.00$ , $0.11$) & ( $0.00$ , $0.18$) \\ \hline
( $2.20$ , $2.83$) & ( $0.09$ , $0.15$) & ( $0.08$ , $0.13$) & ( $0.08$ , $0.15$) \\ \hline
( $2.83$ , $2.83$) & ( $0.04$ , $0.09$) & ( $0.04$ , $0.08$) & ( $0.02$ , $0.10$) \\ \hline
( $3.46$ , $2.83$) & ( $-0.09$ , $0.04$) & ( $-0.09$ , $0.03$) & ( $-0.10$ , $0.03$) \\ \hline
( $4.08$ , $2.83$) & ( $-0.12$ , $-0.02$) & ( $-0.11$ , $-0.03$) & ( $-0.12$ , $-0.04$) \\ \hline
( $1.57$ , $3.46$) & ( $-0.00$ , $0.04$) & ( $-0.00$ , $0.03$) & ( $0.00$ , $0.07$) \\ \hline
( $2.20$ , $3.46$) & ( $0.12$ , $0.02$) & ( $0.12$ , $-0.01$) & ( $0.12$ , $0.04$) \\ \hline
( $2.83$ , $3.46$) & ( $0.09$ , $-0.04$) & ( $0.11$ , $-0.03$) & ( $0.09$ , $-0.02$) \\ \hline
( $3.46$ , $3.46$) & ( $-0.04$ , $-0.09$) & ( $-0.03$ , $-0.08$) & ( $-0.03$ , $-0.09$) \\ \hline
( $4.08$ , $3.46$) & ( $-0.09$ , $-0.15$) & ( $-0.07$ , $-0.14$) & ( $-0.08$ , $-0.15$) \\ \hline
( $1.57$ , $4.08$) & ( $-0.00$ , $-0.05$) & ( $-0.00$ , $-0.02$) & ( $0.00$ , $-0.03$) \\ \hline
( $2.20$ , $4.08$) & ( $0.15$ , $-0.06$) & ( $0.14$ , $-0.05$) & ( $0.15$ , $-0.04$) \\ \hline
( $2.83$ , $4.08$) & ( $0.15$ , $-0.09$) & ( $0.14$ , $-0.10$) & ( $0.15$ , $-0.08$) \\ \hline
( $3.46$ , $4.08$) & ( $0.02$ , $-0.12$) & ( $0.05$ , $-0.11$) & ( $0.04$ , $-0.12$) \\ \hline
( $4.08$ , $4.08$) & ( $-0.06$ , $-0.15$) & ( $-0.06$ , $-0.13$) & ( $-0.04$ , $-0.15$) \\   \hline \hline
    $L^1_{[\frac{\pi}{2},\frac{3\pi}{2}]^2}$ norm error  &  ~             &$\mbox{Error $=0.26$}$~~          &$\mbox{Error $=0.20$}$~~             \\ \hline
\end{tabular}}
\caption{The exact values of the velocity field of the Taylor--Green vortex at $t=1$ (second column) at 25 lattice points (first column), and the values of the computed velocity field at those points, with $\Delta t = 0.02$ (third column), $\Delta t= 0.01$ (fourth column) corresponding to the images depicted in Fig. \ref{figure:new1}.}
\label{tab:2}
\end{table}

\begin{figure}[ht]
\centering
  \begin{subfigure}[ht]{0.3\textwidth}
        \centering
        \includegraphics[width=1\linewidth]{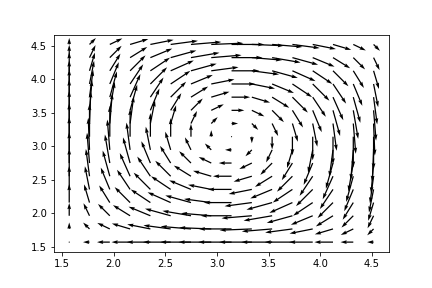}
        \caption{}
    \end{subfigure}
    \begin{subfigure}[ht]{0.3\textwidth}
        \centering
        \includegraphics[width=1\linewidth]{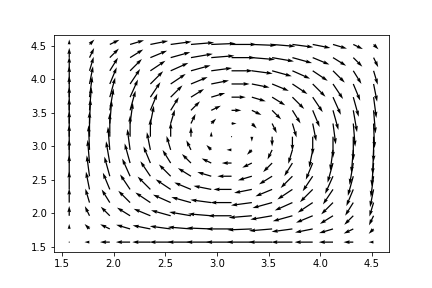}
        \caption{}
        
    \end{subfigure}
        \begin{subfigure}[ht]{0.3\textwidth}
        \centering
        \includegraphics[width=1\linewidth]{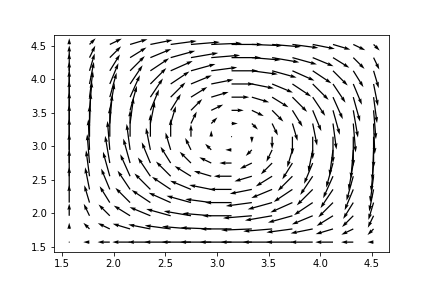}
        \caption{}
        
    \end{subfigure}
    
    \begin{subfigure}[ht]{0.3\textwidth}
        \centering
        \includegraphics[width=1\linewidth]{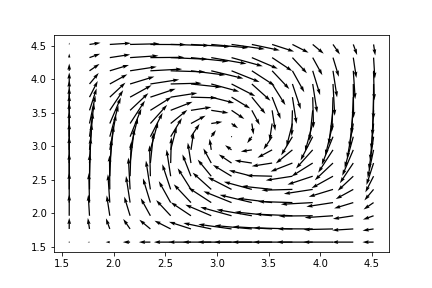}
        \caption{}
    \end{subfigure}
    \begin{subfigure}[ht]{0.3\textwidth}
        \centering
        \includegraphics[width=1\linewidth]{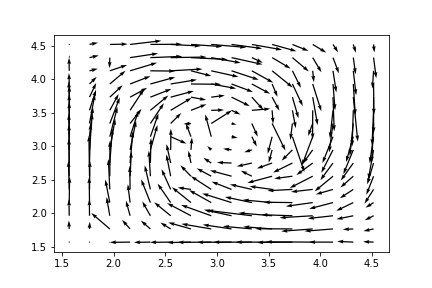}
        \caption{}
       \end{subfigure}
           \begin{subfigure}[ht]{0.3\textwidth}
        \centering
        \includegraphics[width=1\linewidth]{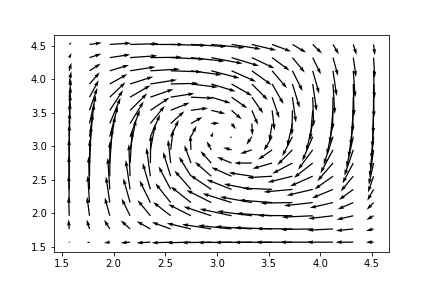}
        \caption{}
        
    \end{subfigure}

    \begin{subfigure}[ht]{0.3\textwidth}
        \centering
        \includegraphics[width=1\linewidth]{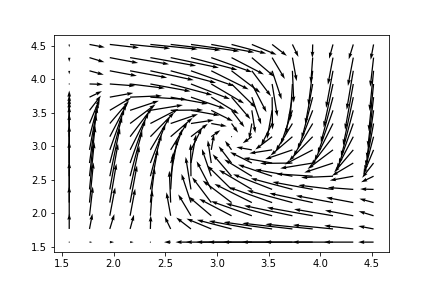}
        \caption{}
    \end{subfigure}
    \begin{subfigure}[ht]{0.3\textwidth}
           \centering
        \includegraphics[width=1\linewidth]{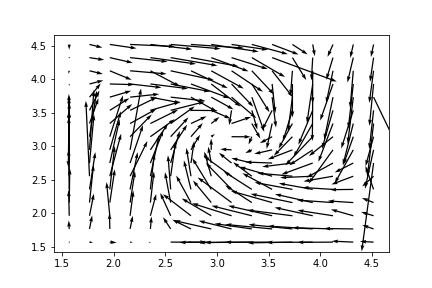}
        \caption{}
    \end{subfigure}
        \begin{subfigure}[ht]{0.3\textwidth}
        \centering
        \includegraphics[width=1\linewidth]{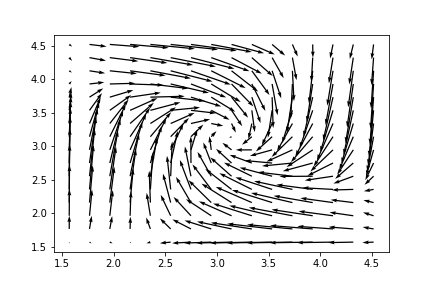}
        \caption{}
        
    \end{subfigure}
    
        \begin{subfigure}[ht]{0.3\textwidth}
        \centering
        \includegraphics[width=1\linewidth]{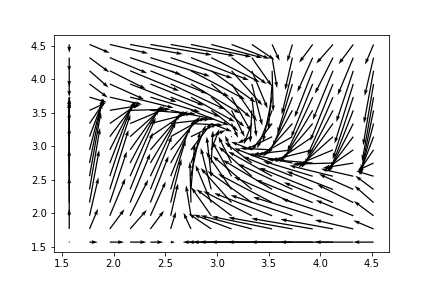}
        \caption{}
    \end{subfigure}
    \begin{subfigure}[ht]{0.3\textwidth}
           \centering
        \includegraphics[width=1\linewidth]{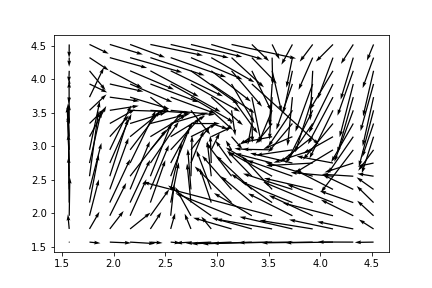}
        \caption{}
    \end{subfigure}
        \begin{subfigure}[ht]{0.3\textwidth}
        \centering
        \includegraphics[width=1\linewidth]{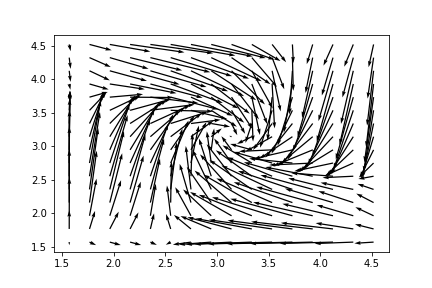}
        \caption{}
        
    \end{subfigure}
    
    \caption{The images in the left column represent the Taylor--Green vortex solution. Those in the right column have been obtained using  the proposed scheme. From the top to the bottom, $t=0,0.3,0.6,1$, respectively.  The images in the left column represent the Taylor--Green vortex solution. Those in the middle column and right column have been obtained using  the proposed scheme with step size $\Delta t=0.02$ and $0.01$, respectively; from the top to the bottom, $t=0, 0.3, 0.6, 1$.}
\label{figure:new1}
\end{figure}

\begin{figure}[ht]
\centering
  \begin{subfigure}[ht]{0.35\textwidth}
        \centering
        \includegraphics[width=1\linewidth]{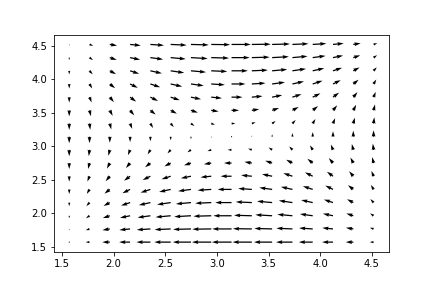}
        \caption{}
    \end{subfigure}
    \begin{subfigure}[ht]{0.35\textwidth}
        \centering
        \includegraphics[width=1\linewidth]{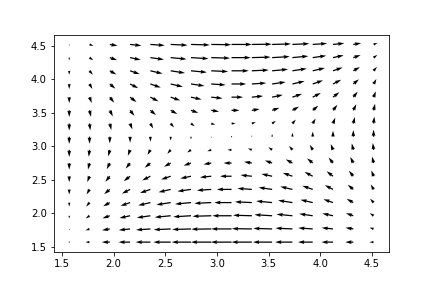}
        \caption{}
        
    \end{subfigure}
    \begin{subfigure}[ht]{0.35\textwidth}
        \centering
        \includegraphics[width=1\linewidth]{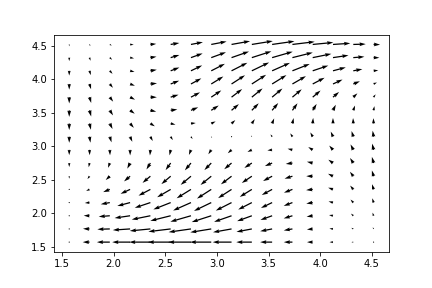}
        \caption{}
    \end{subfigure}
    \begin{subfigure}[ht]{0.35\textwidth}
        \centering
        \includegraphics[width=1\linewidth]{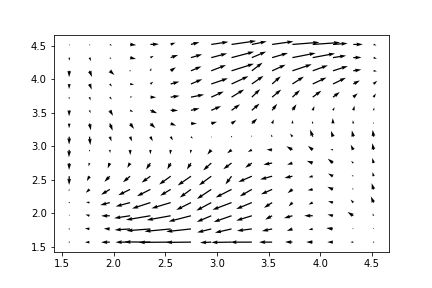}
        \caption{}
        
    \end{subfigure}
    \begin{subfigure}[ht]{0.35\textwidth}
        \centering
        \includegraphics[width=1\linewidth]{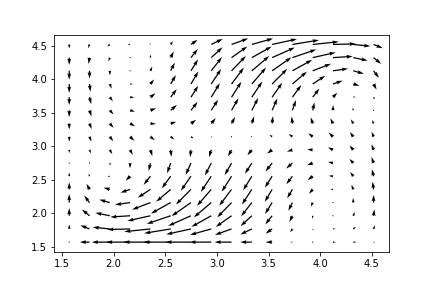}
        \caption{}
    \end{subfigure}
    \begin{subfigure}[ht]{0.35\textwidth}
           \centering
        \includegraphics[width=1\linewidth]{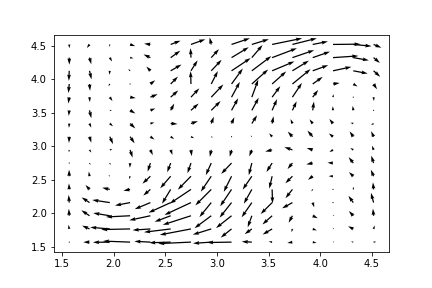}
        \caption{}
    \end{subfigure}
    
        \begin{subfigure}[ht]{0.35\textwidth}
        \centering
        \includegraphics[width=1\linewidth]{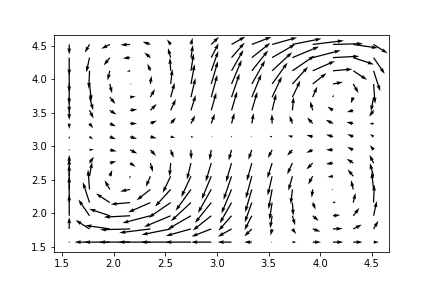}
        \caption{}
    \end{subfigure}
    \begin{subfigure}[ht]{0.35\textwidth}
           \centering
        \includegraphics[width=1\linewidth]{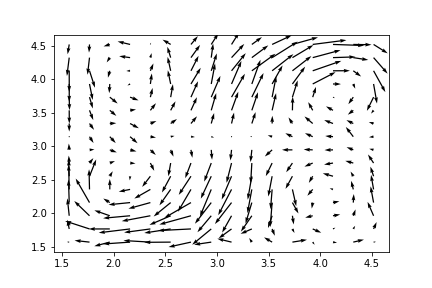}
        \caption{}
    \end{subfigure}
       \caption{The images in the left column represent the DNS solution. Those in the right column have been obtained using  the proposed scheme; from the top to the bottom, $t=0, 0.3, 0.6, 1$, respectively.}
\label{figure:new2}
\end{figure}

\section*{Data Accessibility}
The code for the  simulations is available on Dryad \cite{vortexcode}.

\section*{Acknowledgement}
The authors wish to thank Dan Crisan and Rongchan Zhu for bringing to our attention the papers \cite{constantin2008stochastic} and \cite{zhang2010stochastic,zhang2016stochastic}, and to one of the referees for pointing out to us the  paper \cite{wang2021sharp}. This article is based on work partially supported by the EPSRC Centre for Doctoral Training in Mathematics of Random Systems: Analysis, Modelling and Simulation (EP/S023925/1).

\bibliographystyle{abbrv}
	\bibliography{references}

\end{document}